\documentclass[11pt]{article}
\usepackage[margin=1in]{geometry}
\usepackage{graphicx}
\usepackage{amsmath}
\usepackage{amsthm}
\usepackage{physics}
\usepackage{amssymb}
\usepackage{xcolor}
\usepackage{hyperref}
\usepackage{mathtools}
\usepackage{comment}
\usepackage{subcaption}
\usepackage[normalem]{ulem}
\usepackage{mathrsfs}
\usepackage{enumitem}
\usepackage{authblk}

\theoremstyle{plain}
\newtheorem{theorem}{Theorem}[section]
\newtheorem{lemma}[theorem]{Lemma}
\newtheorem{proposition}[theorem]{Proposition}
\newtheorem{corollary}[theorem]{Corollary}

\theoremstyle{definition}

\theoremstyle{remark}

\newtheorem{claim}[theorem]{Claim}

\numberwithin{equation}{section}

\usepackage{dsfont}
\newcommand{\eps}{\epsilon}
\newcommand{\indicator}[1]{\mathds{1}_{#1}}
\DeclarePairedDelimiter{\bars}{\lvert}{\rvert}
\DeclarePairedDelimiter{\set}{\{}{\}}
\newcommand{\ceil}[1]{\left\lceil#1\right\rceil}
\newcommand{\floor}[1]{\left\lfloor#1\right\rfloor}
\DeclarePairedDelimiter{\parentheses}{(}{)}

\newcommand{\E}{\mathbb{E}}
\renewcommand{\P}{\mathbb{P}}
\newcommand{\R}{\mathbb{R}}
\newcommand{\V}{\mathbb{V}}
\newcommand{\Z}{\mathbb{Z}}

\newcommand{\cC}{\mathcal{C}}
\newcommand{\cE}{\mathcal{E}}
\newcommand{\cF}{\mathcal{F}}
\newcommand{\cG}{\mathcal{G}}
\newcommand{\cL}{\mathcal{L}}
\newcommand{\cP}{\mathcal{P}}

\newcommand{\fm}{\mathfrak{m}}
\newcommand{\fn}{\mathfrak{n}}
\newcommand{\fs}{\mathfrak{s}}
\newcommand{\fv}{\mathfrak{v}}
\newcommand{\sF}{\mathscr{F}}
\newcommand{\sP}{\mathscr{P}}

\newcommand{\eff}{\mathrm{eff}}
\newcommand{\oned}{\mathrm{1D}}
\newcommand{\per}{\mathrm{per}}
\newcommand{\mic}{\mathrm{mic}}
\newcommand{\Cov}{\mathrm{Cov}}
\newcommand{\Side}{\mathrm{Side}}
\newcommand\rect[2]{\mathrm{R}_{#1,#2}}
\newcommand\peierls[1]{\mathfrak{p}_{#1}}
\newcommand\es[1]{\overline{#1}}
\newcommand\xre[2]{X_{#1}#2}
\newcommand{\Phase}[1]{\mathrm{Phase}(#1)}

\DeclarePairedDelimiter{\brackets}{[}{]}

\let\emptyset\varnothing
\let\vec\mathbf

\allowdisplaybreaks

\title{Nematic order in monomer-dimer models of Heilmann and Lieb with lateral attraction in dimensions $d\ge 2$}
\author{Qidong He\thanks{Email address: \url{q.he@northeastern.edu}}}
\date{}
\affil{Department of Mathematics, Northeastern University}

\begin{document}

\maketitle

\begin{abstract}
    We study the $d$-dimensional generalization of Models~II ($d=2$) and V ($d=3$) for liquid crystals introduced by Heilmann--Lieb (1979), which consists of dimers on the hypercubic lattice at chemical potential $\mu$, interacting via a hard-core repulsion and a lateral attraction between parallel dimers with coupling constant $b>0$.
    For $d\in\set{2,3}$, these models were conjectured to exhibit liquid-crystalline order at low temperatures provided that $\mu+2(d-1)b>0$.
    We establish nematic order for all $d\ge 2$ under the additional condition $(4d-6)b>\mu$, which corresponds to the physical regime in which misaligned dimers are rare on scales comparable to the correlation length of the one-dimensional reference system.
    While the treatment of orientational order here follows the approach of our recent work on Model~I, the proof of the absence of translational order relies on a custom implementation of cluster swapping in the sense of Sheffield (2005).
    Consequently, we deduce uniqueness of Gibbs measure and anisotropic exponential decay of correlations within each oriented phase.
\end{abstract}


\section{Introduction}

Liquid crystals are a state of matter intermediate between a liquid and a solid crystal, characterized microscopically by broken rotational symmetry while retaining translation invariance in at least one spatial direction~\cite{de1993physics}.
Deriving this behavior from concrete microscopic interactions is a delicate and significant mathematical challenge, as it requires simultaneous control of two seemingly distinct properties: the stability of global alignment and the absence of full accompanying crystalline order.
Establishing the former does not automatically settle the latter, since an aligned system may still develop positional order in every spatial direction~\cite{bolhuis1997tracing}.
In the mathematical literature, these problems have been addressed for a limited number of two-dimensional models, including nematic order in systems of rods of fixed~\cite{disertori2013nematic} and variable~\cite{ioffe2006entropy} lengths, rectangles~\cite{he2026rigorous}, and interacting dimers~\cite{jauslin2018nematic,he2025extended} on $\Z^2$, as well as columnar order in a system of $2\times 2$ squares on $\Z^2$ at high densities~\cite{hadas2026columnar}.
Comparable results establishing liquid-crystalline order in higher dimensions are much rarer; a notable exception is the proof of uniaxial nematic order in systems of anisotropic plates in $\R^3$~\cite{disertori2020plate}.
We emphasize that all these models restrict the particles to having a \textit{finite} number of orientations, whereas Onsager's famous original prediction~\cite{onsager1949effects} concerns the development of nematic order in systems of elongated particles with \textit{continuous} rotational symmetry.
For models with continuous rotational symmetry, there are a small number of rigorous results concerning Onsager-type free-energy functionals~\cite{fatkullin2005critical,vollmer2017critical} and their derivations in the mean-field~\cite{bachmann2017mean} and van der Waals--Kac~\cite{erignoux2020nematic} limits.

In this paper, we continue our study of the lattice models for liquid crystals proposed by Heilmann--Lieb in 1979~\cite{heilmann1979lattice}, focusing specifically on Models~II and~V.
Both models feature the same types of microscopic interactions and differ only in dimension: Model~II is two-dimensional, while Model~V is three-dimensional.
Particles are represented by dimers (edges) of a hypercubic lattice $\Z^d$, and the interactions between particles consist of a hard-core repulsion and a lateral attractive interaction between adjacent, parallel dimers.
With $\mu\in\R$ and $b>0$ denoting the (dimer) chemical potential and the strength of the attractive interaction, Heilmann--Lieb proved by reflection-positivity methods that, for $d\in\set{2,3}$, both models exhibit orientational order at low temperatures provided that $\mu+2(d-1)b>0$.
However, as with the other models presented in their paper, they stopped short of ascertaining the status of translational symmetry within each oriented phase: for instance, \textit{nematic} liquid-crystalline order (depicted in Figure~\ref{fig: model V} for Model~V), as we establish in this paper for both models as well as their higher-dimensional generalizations, would follow from a proof that no translational symmetry breaking occurs whatsoever.

\begin{figure}[t]
    \centering
    \includegraphics[width=0.5\textwidth]{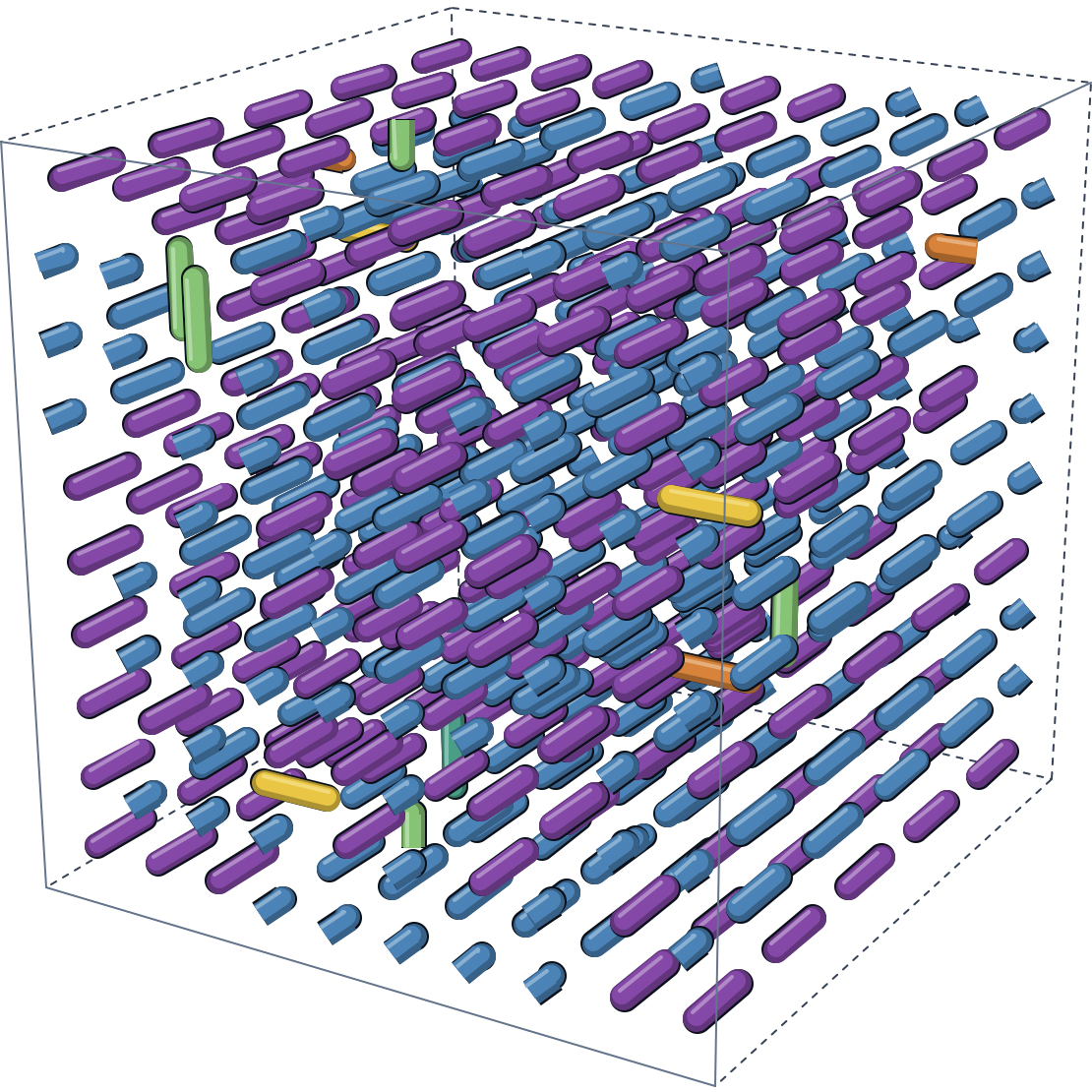}
    \caption{A typical configuration of Model~V on a $10\times10\times10$ torus with parameters $\beta=1.8$, $\mu=0.2$, and $b=0.3$. Nematic order is evident.}
    \label{fig: model V}
\end{figure}

Whereas Model~I has been the subject of a number of subsequent mathematical works~\cite{alberici2016cluster,jauslin2018nematic,he2025extended}, the remaining models, to our knowledge, have seen virtually no progress since the publication of Heilmann--Lieb's original paper~\cite{heilmann1979lattice}.
For Models~II and~V, specifically, we speculate that this is in large part due to the perceived complexity of the attractive interaction therein, which is built around pairs of \textit{non-collinear} but \textit{parallel} dimers and raises the question of whether the multi-dimensionality can be effectively tamed using existing mathematical tools.
In this paper, we give an \textit{affirmative} answer to this question in the regime $(4d-6)b>\mu$, though this requires a methodological shift from our previous work on Model~I~\cite{he2025extended}, as we now explain.

In~\cite{he2025extended}, the proof of the absence of translational order relies on the method of disagreement percolation~\cite{van1993uniqueness,van1994disagreement}, which requires considering an extended connectivity on the edge set of $\Z^2$ under which Model~I is a Markov random field.
Since the attractive interaction in Model~I is between \textit{collinear} dimers, the extended connectivity reduces to an effectively one-dimensional one once misalignment becomes sufficiently rare.
Indeed, the so-called \textit{disagreement components} that ensue are simply one-dimensional paths on $\Z^2$, which are easy to control.
This picture breaks down completely for Models~II and~V.
Due to the nature of the attractive interaction, any connectivity on the edge set of $\Z^d$ that makes these models into Markov random fields must connect pairs of parallel (non-collinear) edges.
Hence, even when misalignment is rare, disagreement components are still sprawling, multi-dimensional objects that are difficult to control.

Using the method of cluster swapping~\cite{sheffield2005random,peled2020three}, we are able to overcome this difficulty.
Cluster swapping, as introduced by Sheffield~\cite{sheffield2005random}, is a generalization of disagreement percolation that allows one to consider more flexible \textit{clusters} than disagreement components with respect to a particular connectivity.
It turns out that, for Models~II and~V, the structure of the attractive interaction permits an especially convenient choice of clusters (see Section~\ref{sec:cluster swapping}) such that, provided that misalignment is sufficiently rare, the clusters reduce to one-dimensional paths just as the disagreement components did for Model~I.
Therefore, the move from disagreement percolation to cluster swapping is not merely a technical refinement but an essential step in overcoming the difficulties posed by the multi-dimensional attractive interaction in Models~II and~V.

Apart from the transition to cluster swapping, the proof of nematic order in this paper follows the same strategy as that of~\cite{he2025extended}.
Namely, we start from the same microscopic characterization of orientational order as used by Heilmann--Lieb~\cite{heilmann1979lattice} and strengthen it into a much sharper quantification of the rarity of misaligned dimers in each oriented phase.
We then define the clusters and use the sharper estimate to prove the strong percolation of \textit{sealed boxes} which confine passing clusters to one-dimensional columns, making them vulnerable to termination.
Thus, we verify the assumption for cluster swapping and conclude from it uniqueness of Gibbs measure and exponential decay of correlations within each oriented phase.

We mention that the additional condition $(4d-6)b>\mu$ plays the same role as the condition $3a>\lambda$ did in our work on Model~I~\cite{he2025extended} (the dimer chemical potential there is denoted by $\lambda$).
Namely, it separates the length scale on which vacancies are frequent from that on which the existence of misaligned dimers is non-negligible.
What happens when this separation of scales is no longer valid is an interesting open question.
We refer the reader to~\cite[Section~6]{he2025extended} for a related discussion in the context of Model~I.

\paragraph{Model and main results}

We now introduce the model and state the main results.

Let $d\ge 2$ and $\V$ be the edge set of the hypercubic lattice $\Z^d$.
We identify each element of $\V$ with its midpoint in $\R^d$.
A \textbf{configuration} is a function $\sigma:\V\to\set{0,1}$ such that $\sigma(e)\sigma(e')=0$ whenever $e,e'\in\V$ are distinct and incident to the same vertex of $\Z^d$.
We identify each configuration $\sigma$ with the set $\sigma^{-1}(\set{1})$ of pairwise non-incident edges of $\Z^d$.
Let $\Omega$ be the set of all configurations on $\V$.

Let $\beta>0$ be the inverse temperature, $\mu\in\R$ be the chemical potential, and $b>0$ be the coupling constant.
The formal Hamiltonian of the $d$-dimensional monomer-dimer model with lateral attraction, as Heilmann--Lieb introduced it in~\cite{heilmann1979lattice} (for $d\in\set{2,3}$), is
\begin{equation}
    \label{eqn:formal Hamiltonian}
    H(\sigma)= -\mu\abs{\sigma}-b\,\sum_{e\in\V}\indicator{\begin{subarray}{l}\text{there exist parallel dimers }e_1,e_2\in\sigma\\\text{each perpendicular and incident to }e\end{subarray}},
\end{equation}
where $\abs{\cdot}$ denotes the cardinality.
The formal Hamiltonian~\eqref{eqn:formal Hamiltonian} is adapted to finite volumes in the usual way and gives rise to finite-volume Gibbs measures as well as infinite-volume ones via the usual Dobrushin--Lanford--Ruelle formalism.
A detailed review of the definitions and constructions referenced above is given in Section~\ref{sec:the Hamiltonian}.

Our main result is as follows; see Figure~\ref{fig: model V} for an illustration in the case $d=3$.

\begin{theorem}[Nematic order]
    \label{thm:main}
    Fix $d\ge 2$, $\mu\in\R$, and $b>0$ such that $\mu+2(d-1)b>0$ and $(4d-6)b>\mu$.
    There exist constants $\beta_0,C,c>0$ and $\fn\ge 4(6d-5)[2^{d+1}(4d-5)+1]$ (see~\eqref{eqn: microscopic length scale n} for an exact expression) such that, for all $\beta\ge\beta_0$, the $d$-dimensional monomer-dimer model with lateral attraction admits a unique Gibbs measure $\P_{x_1}$ such that:
    \begin{enumerate}
        \item \label{itm:main - invariance and extremality} $\P_{x_1}$ is $\Z^d$-invariant and extremal.
        \item \label{itm:main - dimer probabilities} For all $e\in\V$,
        \begin{equation}
            \P_{x_1}(\sigma(e)=1)=
            \begin{cases}
                \frac{1}{2}-\order{e^{-\frac{1}{2}\beta [\mu+2(d-1)b]}} & \text{if }e\text{ is $x_1$-oriented} \\ 
                \order{e^{-\beta(1-1/\fn)(4d-5)b}} & \text{otherwise}
            \end{cases}.
        \end{equation}
        \item \label{itm:main - decay of correlations} Let $A,B\subset\V$ and $f,g:\Omega\to[-1,1]$ be such that $f$ depends only on the restriction of $\sigma$ to $A$ and $g$ only on the restriction of $\sigma$ to $B$.
        Then,
        \begin{equation}
            \abs{\Cov_{\P_{x_1}}(f,g)}\le \sum_{u\in A}\sup_{v\in B}\alpha_{x_1}(u,v),
        \end{equation}
        where, defining $\ell_0:=e^{\frac{1}{2}\beta[\mu+4(d-1)b]}$,
        \begin{equation}
            \alpha_{x_1}(u,v):=C\exp\set[\bigg]{-c\brackets[\bigg]{\ell_0^{-1}\abs{u_1-v_1}+\sum_{i=2}^d\abs{u_i-v_i}}}.
        \end{equation}
    \end{enumerate}
    Let $\varphi:\R^d\to\R^d$ be the coordinate permutation $(x_1,\dots,x_{d-1},x_d)\mapsto(x_d,x_1,\dots,x_{d-1})$.
    For each $2\le j\le d$, the push-forward $\P_{x_j}:=\varphi^{j-1}\P_{x_1}$ is another Gibbs measure of the model that satisfies analogous properties.
    Furthermore, every periodic Gibbs measure of the model is a convex combination of the $\P_{x_i}$, $1\le i\le d$.
\end{theorem}

\paragraph{Organization of the paper}

In Section~\ref{sec:preliminaries}, we formally define the monomer-dimer model and review the reflection-positivity property, the chessboard estimate, and the framework of strongly percolating sets~\cite[Section~6]{hadas2026columnar}.
In Section~\ref{sec:one-dimensional systems}, we analyze the one-dimensional reference model using the transfer-matrix method and derive the basic estimates used later in the proof. 
In Section~\ref{sec: orientational symmetry breaking}, we first formulate the Heilmann--Lieb characterization of the oriented phases in the strong-percolation framework, and then strengthen it to a near-optimal quantitative characterization of the rarity of misaligned dimers in each phase.
In Section~\ref{sec: absence of translational symmetry breaking}, we construct sealed boxes, prove their strong percolation on the mesoscopic scale, and implement cluster swapping to prove uniqueness of Gibbs measure and exponential decay of correlations within each phase.
Finally, in Section~\ref{sec:microscopic characterization of nematic order}, we characterize the phases microscopically by estimating the densities of aligned and misaligned dimers.

\section{Preliminaries}
\label{sec:preliminaries}

\subsection{Basic definitions}
\label{sec:basic definitions}

\paragraph{Graphical constructions}

For $d\ge 2$, let $\Z^d$ be the $d$-dimensional hypercubic lattice.
Denote by $L(\Z^d)$ the \textbf{line graph} of $\Z^d$, which has vertices corresponding to the edges of $\Z^d$, and where two vertices are connected by an edge if and only if their corresponding edges in $\Z^d$ share a vertex.
Denote the vertex set of $L(\Z^d)$ by $\V$; the dimension $d$ underlying the latter will always be clear from the context.
We identify the elements of $\V$ with the midpoints of the corresponding edges of $\Z^d$.
Thus, each $e\in\V$ may be viewed both as an edge of $\Z^d$ and as a vertex of $L(\Z^d)$.
Unless otherwise noted, we adopt the former viewpoint from here on, so that the terms ``vertex'' and ``edge'' refer to a vertex and an edge of $\Z^d$, respectively.

We will consider two connectivities on $\Z^d$ itself.
The first, denoted by the symbol $\Box$, is the standard connectivity on $\Z^d$ induced by the $1$-norm on $\R^d$.
The second, denoted by the symbol $\boxtimes$, is induced by the sup-norm on $\R^d$ and is particularly useful for characterizing the connectivity of boundaries of $\Box$-connected subsets of $\Z^d$; see~\cite{timar2013boundary}.

\paragraph{Boxes}

In this paper, the term \textbf{box} refers to a closed, axis-aligned rectangular region in $\R^d$ with corners at the geometric centers of the cells of $\Z^d$: for $v_1,\dots,v_d\in 1/2+\Z$ and $L_1,\dots,L_d\in\Z_{>0}$, denote
\begin{equation}
    \rect{L_1\times \dots\times L_d}{(v_1,\dots,v_d)}:=[v_1,v_1+L_1]\times \dots \times [v_d,v_d+L_d]\subset\R^d.
\end{equation}
In particular, we write
\begin{equation}
      \mathrm{R}_{L_1\times \dots\times L_d}:=\rect{L_1\times \dots\times L_d}{(-1/2,\dots,-1/2)}.
\end{equation}
A box is said to be \textbf{even} if its side lengths are all even.

\paragraph{Configurations}

A \textbf{dimer} is an edge $e\in\V$.
Its \textbf{orientation} is the unique label $o\in\set{x_1,\dots,x_d}$ such that $e$ is parallel to the $o$-axis, and we say that $e$ is \textbf{$o$-oriented} in that case.
Two dimers $e,e'\in\V$ are \textbf{incident} if they intersect at a vertex of $\Z^d$.
A \textbf{configuration} on $\V$ is a function $\sigma:\V\rightarrow\set{0,1}$ with the property that no two dimers in $\sigma^{-1}(\set{1})$ are incident.
We identify the function $\sigma$ with the set $\sigma^{-1}(\set{1})$.
Denote the set of all configurations on $\V$ by $\Omega$.
Equip $\Omega$ with the subspace topology inherited from the product topology on $\set{0,1}^\V$ and with the corresponding Borel measurable structure; denote the resulting $\sigma$-algebra by $\sF$.

Let $\Lambda\subset\R^d$ be a bounded region.
We introduce restricted sets of configurations given by various boundary conditions on $\Lambda$ (or, more precisely, on $\V\cap\Lambda$):
\begin{enumerate}
    \item \emph{Prescribed boundary conditions}.
    Let $\rho\in\Omega$.
    Define the set of configurations in $\Lambda$ with $\rho$-boundary condition as
    \begin{equation}
        \Omega^\rho_\Lambda:=\set{\sigma\in\Omega\mid\text{for all }e\in\V\setminus\Lambda,\,\sigma(e)=\rho(e)}.
    \end{equation}
    \item \emph{Periodic boundary conditions}.
    Suppose that $\Lambda$ is a box.
    Define the set of $\Lambda$-periodic configurations as
    \begin{equation}
        \Omega^\per_\Lambda:=\set[\big]{\sigma\in\Omega\mid\text{for all }e\in\V\text{ and }1\le i\le d,\,
    \sigma(e)=\sigma(e+\Side_i(\Lambda)\vec{e}_i)},
    \end{equation}
    where $\Side_i(\Lambda)$ is the $i$th side length of $\Lambda$ and $\vec{e}_i$ is the $i$th standard basis vector of $\R^d$.
\end{enumerate}

Finally, given a set $A\subset\R^d$ and a measurable function $f:\Omega\to\R$, we say that $f$ is \textbf{$A$-local} if the value of $f$ depends only on the restriction of a configuration to $A\cap\V$.

\subsection{The Hamiltonian and Gibbs measures}
\label{sec:the Hamiltonian}

Let $\sigma\in\Omega$, and recall that we have identified $\sigma\equiv\sigma^{-1}(\set{1})$.
An edge $e\in\V$ is a \textbf{link} of $\sigma$ if its endpoints are incident to parallel dimers in $\sigma$ which are both perpendicular to $e$.
Let $\beta>0$ be the inverse temperature, $\mu\in\R$ be the (dimer) chemical potential, and $b>0$ be the strength of attraction associated with each link.
In the remainder of the paper, we will consider $\mu$ and $b$ to be fixed such that $\mu+2(d-1)b>0$ and $(4d-6)b>\mu$, and suppress them from the notation.
The formal Hamiltonians of the monomer-dimer models studied in this paper, referred to as Models II ($d=2$) and V ($d=3$) by Heilmann--Lieb in~\cite{heilmann1979lattice}, have the common form
\begin{equation}
    H(\sigma):=-\mu\abs{\sigma}-b\sum_{e\in\V}\indicator{e\text{ is a link of }\sigma},
\end{equation}
where $\abs{\sigma}$ is the number of dimers in the configuration $\sigma$.

In this paper, we study these models in the perturbative regime $\beta\gg1$.
To exploit the largeness of $\beta$, it is advantageous to shift the Hamiltonian by a constant multiple of the volume, as exemplified in the following formal rewriting, whose adaptation to finite volumes is straightforward:
\begin{equation}
    \label{eqn:shifted Hamiltonian monomer-dimer}
    \begin{multlined}
        \frac{1}{2}\brackets[\big]{\mu+2(d-1)b}\bars[\big]{\Z^d}+H(\sigma)
        \\
        =\frac{1}{2}\brackets[\big]{\mu+2(d-1)b}(\bars[\big]{\Z^d}-2\abs{\sigma})
        +\frac{1}{2}b\sum_{e\in\V}\parentheses[\Bigg]{\sum_{e'\in\sigma}\indicator{e'\cap e\ne\emptyset,\,e'\perp e}-2\cdot\indicator{e\text{ is a link of }\sigma}}.
    \end{multlined}
\end{equation}
Define a \textbf{broken link} as an edge $e\in\V$ such that either
\begin{enumerate}
    \item \label{itm: broken link multiplicity 1} exactly one endpoint of $e$ is incident to a dimer $e'\in\sigma$ perpendicular to $e$, or
    \item \label{itm: broken link multiplicity 2} the edge $e$ is incident to two (distinct) nonparallel dimers $e_1,e_2\in\sigma$, both perpendicular to $e$.
\end{enumerate}
Define the \textbf{multiplicity} of the broken link $e\in\V$ as
\begin{equation}
    m_\sigma(e):=\begin{cases}
        1 & \text{in Case~\ref{itm: broken link multiplicity 1}} \\
        2 & \text{in Case~\ref{itm: broken link multiplicity 2}}
    \end{cases},
\end{equation}
and set $m_\sigma(e):=0$ if $e$ is not a broken link of $\sigma$.
Thus, the terms on the RHS of~\eqref{eqn:shifted Hamiltonian monomer-dimer} have the following intuitive interpretations:
\begin{enumerate}
    \item $\abs{\Z^d}-2\abs{\sigma}$ is the number of vacant vertices, or \textbf{vacancies} for short, which are vertices of $\Z^d$ not incident to any dimer in $\sigma$.
    \item The summand $\sum_{e'\in\sigma}\indicator{e'\cap e\ne\emptyset,\,e'\perp e}-2\cdot\indicator{e\text{ is a link of }\sigma}$ is nonzero if and only if $e\in\V$ is a broken link of $\sigma$, and, if so, its value is precisely the multiplicity of that broken link.
\end{enumerate}

For clarity, we recast~\eqref{eqn:shifted Hamiltonian monomer-dimer} systematically using the standard notion of potentials (see~\cite[Section 6.3.2]{friedli2017statistical} or~\cite[(2.2)]{georgii2011gibbs}).
Let $B\subset\V$ be bounded (or $B\Subset\V$ for short).
\begin{enumerate}
    \item If $B=\set{e\mid e\text{ is incident to }v}$ for some vertex $v\in\Z^d$, define the $2d$-body vacancy potential
    \begin{equation}
        \label{eqn:vacancy potential}
        \Phi_B(\sigma):=\frac{1}{2}\brackets[\big]{\mu+2(d-1)b}\indicator{v\text{ is a vacancy of }\sigma},
    \end{equation}
    which outputs $\brackets[\big]{\mu+2(d-1)b}/2$ if and only if $v$ is a vacancy of $\sigma$.
    \item If $B=\set{e'\in\V\mid e'\cap e\ne\emptyset}$ for some edge $e\in\V$ (so that, in particular, $e\in B$), define the $(4d-1)$-body broken-link potential
    \begin{equation}
        \label{eqn:broken-link potential}
        \Phi_B(\sigma):=\frac{m_\sigma(e)}{2}b.
    \end{equation}
    \item Otherwise, set $\Phi_B(\sigma):=0$.
\end{enumerate}
Thus, the shifted formal Hamiltonian can be expressed as
\begin{equation}
    \sum_{B\Subset\V}\Phi_B(\sigma),
\end{equation}
whose adaptation to finite volumes is standard (see~\cite[(6.24)]{friedli2017statistical} or~\cite[(2.3)]{georgii2011gibbs}): if $\Lambda\subset\R^d$ is bounded and $\#$ denotes a boundary condition for $\Lambda$, then the finite-volume Hamiltonian for configurations in $\Lambda$ with $\#$-boundary condition is
\begin{equation}
    \label{eqn:finite-volume Hamiltonian}
    H^\#_{\Lambda}(\sigma):=\sum_{\substack{B\Subset\V\\ B\cap\Lambda\ne\emptyset}}\Phi_B(\sigma).
\end{equation}
As an exception, we note that if $\#$ denotes the periodic boundary condition, then the sum over $B\Subset\V$ is instead restricted to a complete set of representatives of the equivalence classes of all $B\Subset\V$ under translations generated by $\set{\Side_i(\Lambda)\vec{e}_i\mid 1\le i\le d}$.

\paragraph{Gibbs measures}

Let $\Lambda\subset\R^d$ be bounded and $\#$ denote a boundary condition on $\Lambda$.
The weight of a configuration $\sigma\in\Omega^\#_\Lambda$ is defined as
\begin{equation}
    w^\#_{\Lambda;\beta}(\sigma):=e^{-\beta H^\#_{\Lambda}(\sigma)}.
\end{equation}
Hence, each vacancy carries a penalty factor of $e^{-\frac{1}{2}\beta[\mu+2(d-1)b]}$, whereas each broken link carries a penalty factor of $e^{-\frac{1}{2}\beta b m_\sigma}$.
The finite-volume Gibbs measure on $\Lambda$ with $\#$-boundary condition is given by
\begin{equation}
    \P^\#_{\Lambda;\beta}(\sigma):=\frac{w^\#_{\Lambda;\beta}(\sigma)}{Z^\#_{\Lambda;\beta}}\quad\text{for }\sigma\in\Omega^\#_\Lambda,
\end{equation}
where $Z^\#_{\Lambda;\beta}:=\sum_{\sigma\in\Omega^\#_\Lambda}w^\#_{\Lambda;\beta}(\sigma)$ is the partition function.
Furthermore, given $E\subset\Omega$, we define the $E$-constrained partition function as
\begin{equation}
    Z^\#_{\Lambda;\beta}(E):=\sum_{\sigma\in\Omega^\#_\Lambda\cap E}w^\#_{\Lambda;\beta}(\sigma).
\end{equation}
Lastly, we define infinite-volume Gibbs measures via the standard Dobrushin--Lanford--Ruelle formalism; see~\cite[Section 6.2.1]{friedli2017statistical} or~\cite[(2.9)]{georgii2011gibbs} for details.

\subsection{Reflection positivity and the chessboard estimate}

In this subsection, we formulate the reflection-positivity of the monomer-dimer models and use it to set up the chessboard estimate.
We will not give any proof below because the results presented here are direct adaptations of those in~\cite[Section~3]{hadas2026columnar} and~\cite[Section~2.3]{he2025extended} and the proofs transfer with little modification.

We begin with some notation.
Let $R=\rect{L_1\times \dots\times L_d}{z}$ be a box.
Define the \textbf{grid} of $R$ as $G^R:=(z_1+L_1\Z)\times\dots\times(z_d+L_d\Z)$, and denote its origin-shifted version by $\cL^R:=L_1\Z\times\dots\times L_d\Z$.
Let $T^R$ be the group generated by the reflections of $R$ across the axis-aligned hyperplanes that pass through points in $G^R$.
Thus, for each $v\in G^R$, there exists a unique isometry in $T^R$ that maps $R$ to $\rect{L_1\times \dots\times L_d}{v}$; we denote this isometry by $\tau_{R,v}$.
Given $f:\Omega\to\R$ and $\tau\in T^R$, define $(\tau f)(\sigma):=f(\sigma\circ \tau)$.
Given a Gibbs measure $\P$ and an isometry $\tau:\R^d\to\R^d$, denote by $\tau\P$ the push-forward of $\P$ under $\tau$.
Given $z\in\Z^d$, let $\eta_{z}:\R^d\to\R^d$ denote the translation $v\mapsto v+z$.

It turns out that the monomer-dimer models considered in this paper are \emph{reflection-positive} with respect to reflections across the axis-aligned hyperplanes bisecting edges of $\Z^d$~\cite[Section~3]{heilmann1979lattice}; for a proof of this fact in the setting of an analogously shifted Hamiltonian for the only other monomer-dimer model introduced in~\cite{heilmann1979lattice}, see~\cite[Proof of Lemma~2.2]{he2025extended}.

\begin{lemma}[Reflection positivity]
    Let $R=\rect{L_1\times \dots\times L_d}{z}$ be a box and $f:\Omega\to\R$ be $R$-local.
    If
    \begin{equation}
        (\Lambda,\tau)\in
        \set[\big]{(\mathrm{R}_{L_1,\dots,L_{i-1},2L_i,L_{i+1},\dots,L_d},\tau_{R,z+L_i\vec{e}_i})\mid 1\le i\le d},
    \end{equation}
    then
    \begin{equation}
        \P_{\Lambda;\beta}^\per(f\cdot\tau f)\ge 0.
    \end{equation}
\end{lemma}

If $\Lambda$ is a box such that $2\Side_i(R)\mid\Side_i(\Lambda)$ for all $1\le i\le d$, then we say that $R$ is a \textbf{block} of $\Lambda$.
In that case, set $T^R_\Lambda:=T^R/\cL^\Lambda$.
Given an $R$-local function $f:\Omega\to\R$, define its \textbf{$(R,\Lambda)$-chessboard seminorm} by 
\begin{equation}
    \norm{f}_{R\mid\Lambda}:=\P_{\Lambda;\beta}^\per\parentheses[\Bigg]{\prod_{\tau\in T_\Lambda^R}\tau f}^{1/\abs{T_\Lambda^R}};
\end{equation}
the name ``seminorm'' is justified by the fact that $\norm{\cdot}_{R\mid\Lambda}$ indeed satisfies the standard properties of a seminorm~\cite[Proposition 3.3]{hadas2026columnar}.
A standard consequence of reflection-positivity is the \emph{chessboard estimate}, which asserts that the expectation of a product of local functions of a particular form under $\P_{\Lambda;\beta}^\per$ factorizes as an upper bound into the product of their chessboard seminorms.

It is straightforward to extend the chessboard estimate to apply to finite products of \textit{periodic} Gibbs measures of the monomer-dimer model by adapting the argument of Hadas--Peled~\cite[Proof of Proposition 3.8]{hadas2026columnar}, which has been done in~\cite{he2025extended} for a monomer-dimer model;
here, we say that a Gibbs measure is \textbf{periodic} if it is invariant under translations by a full-rank sublattice of $\Z^d$.
The results are as follows.

Let $k\ge 1$.
An element $(\sigma_i)_{i=1}^k\in\Omega^k$ may be thought of as a joint configuration that naturally specifies a mapping $\V\to\set{0,1}^k$ via $e\mapsto(\sigma_i(e))_{i=1}^k$.
Thus, given a box $R$, the notion of $R$-locality and transformation under each reflection $\tau\in T^R$ extend naturally to functions $\Omega^k\rightarrow\R$.
We equip $\Omega^k$ with the product $\sigma$-algebra $\sF^{\otimes k}$, and use the symbol $\otimes$ to denote the product of probability measures on $(\Omega,\sF)$.

Let $\Lambda$ be a box, $R$ be a block of $\Lambda$, and $f:\Omega^k\to\R$ be $R$-local.
Define
\begin{equation}
    \norm{f}_{R\mid\Lambda}:=(\otimes_{i=1}^k \P_{\Lambda;\beta}^\per)\parentheses[\Bigg]{\prod_{\tau\in T_\Lambda^R}\tau f}^{1/\abs{T_\Lambda^R}}.
\end{equation}
We note that the product measure $\otimes_{i=1}^k\P_{\Lambda;\beta}^\per$ remains reflection-positive with respect to reflections across the axis-aligned hyperplanes bisecting edges of $\Z^d$, so the quantity above satisfies the same seminorm properties as before.
Define the \textbf{$R$-chessboard seminorm} of $f$ by
\begin{equation}
    \label{eqn:infinite-volume chessboard seminorm}
    \norm{f}_R:=\limsup_{n\to\infty}\norm{f}_{R\mid\mathrm{R}_{n!\times \dots\times n!}}.
\end{equation}

\begin{proposition}[Infinite-volume chessboard estimate]
    \label{prop: infinite-volume chessboard estimate}
    Let $R$ be a box. 
    Let $A\subset T^R$ be finite and $(f_\tau)_{\tau\in A}$ be $R$-local functions from $\Omega^k$ to $\R$.
    Then, given periodic Gibbs measures $(\P_i)_{i=1}^k$,
    \begin{equation}
        (\otimes_{i=1}^k \P_i)\parentheses[\Bigg]{\prod_{\tau\in A}\tau f_\tau}\le\prod_{\tau\in A}\norm{f_\tau}_R.
    \end{equation}
\end{proposition}

\begin{proposition}[Recursive chessboard estimate in infinite volume]
    \label{prop: infinite-volume recursive chessboard estimate}
    Let $R$ and $S$ be boxes, and suppose that the corners of $S$ are in $G^R$.
    Let $A\subset T^R$ be such that $\bigcup_{\tau\in A}\tau R\subset S$ and $(f_\tau)_{\tau\in A}$ be $R$-local functions.
    Then,
    \begin{equation}
        \norm{\prod_{\tau\in A}\tau f_\tau}_{S}\le\prod_{\tau\in A}\norm{f_\tau}_{R}.
    \end{equation}
\end{proposition}

\subsection{Strongly percolating sets}
\label{sec:strongly percolating sets}

A recurrent theme in our analysis is to take a local property of configurations and quantify its prevalence in typical configurations via percolation-type arguments.
In this section, we import the notion of strongly percolating sets and their properties from~\cite[Section 6]{hadas2026columnar}, which are well-suited for this purpose.
While the results in~\cite[Section 6]{hadas2026columnar} are formulated in two dimensions, they extend naturally to all dimensions $d\ge 2$, so we will not belabor their proofs here.

We start with some notation.
Let $E\subset\Omega$ be an event, which should be thought of as a local property of configurations.
Define the random set $\es{E}:\Omega\to 2^{\Z^d}$ by 
\begin{equation}
    \es{E}(\sigma):=\set{z\in\Z^d\mid\sigma\in\eta_{z}E},
\end{equation}
where we recall that $\eta_{z}$ denotes the translation by $z$.
Thus, $\es{E}$ is the random set of locations at which the local property $E$ holds.
Given $L_1,\dots,L_d\in\Z_{>0}$, define the random set
\begin{equation}
    \xre{L_1\times \dots \times L_d}{\es{E}}:=\set{z\in\Z^d\mid(L_1z_1,\dots,L_dz_d)\in\es{E}},
\end{equation}
which tracks the manifestation of the local property $E$ on the scale of the rectangular grid $L_1\Z\times \dots \times L_d\Z$.
The above definitions generalize naturally to \emph{joint} properties of multiple, independently sampled configurations, i.e., events $E\subset\Omega^k$, $k\ge 2$.

We now define strongly percolating sets.
Let $\eps_d$ be a small constant whose value we do not determine; see~\cite[Proof of Lemma~6.2]{hadas2026columnar} for its instance in two dimensions.
Namely, all the results in this subsection should be understood to hold for a sufficiently small constant $\eps_d>0$ which depends only on the dimension $d\ge 2$.
Let $\eps\ge 0$.
A random set $B\subset\Z^d$ on a probability space $(\Omega,\sF,\P)$ is \textbf{$\eps$-rare} if, for all finite $A\subset\Z^d$, $\P(A\subset B)\le\eps^{\abs{A}}$.
A random set $\Psi\subset\Z^d$ on $(\Omega,\sF,\P)$ is \textbf{$\eps$-strongly percolating} if either $\eps\ge\eps_d$ or there exists an $\eps$-rare set $B$ such that $\Psi$ $\P$-almost surely contains an infinite $\Box$-component of $\Z^d\setminus B$.
Define
\begin{equation}
    \label{eqn: strong percolation parameter}
    \peierls{\P}(\Psi):=\inf\set{\eps\ge 0\mid\Psi\text{ is }\eps\text{-strongly percolating}}.
\end{equation}
We mention here that the quantity $\peierls{\P}(\Psi)$ may be thought of as a measure of the size of $\Psi$, with smaller values indicating larger random sets; see~\cite[Section 6.1]{hadas2026columnar} for further discussion.

\begin{lemma}[{\cite[Lemma 6.2]{hadas2026columnar}}]
    \label{lem:Peierls argument}
    If $\eps<\eps_d$ and $B$ is $\eps$-rare with respect to a probability measure $\P$, then $\Z^d\setminus B$ $\P$-almost surely contains a unique infinite $\Box$-component $I$.
    Moreover, each $\boxtimes$-component of $\Z^d\setminus I$ is finite.
\end{lemma}

\begin{lemma}[{\cite[Corollary 6.3]{hadas2026columnar}}]
    \label{lem:strong percolation of complement of rare set}
    If $\eps\ge 0$ and $B$ is $\eps$-rare with respect to a probability measure $\P$, then $\Z^d\setminus B$ is $\eps$-strongly percolating, and in particular $\peierls{\P}(\Z^d\setminus B)\le\eps$.
\end{lemma}

\begin{lemma}[{\cite[Lemma 6.4]{hadas2026columnar}}]
    \label{lem:quantitative Peierls argument}
    Let $u\in\Z^d$, $\ell\ge 0$, and $\Psi\subset\Z^d$ be a random set on a probability space $(\Omega,\sF,\P)$.
    Let $E$ be the event that there exists a $\boxtimes$-path in $\Z^d\setminus\Psi$ starting at $u$ and ending at some point in $\set{v\in\Z^d\mid\norm{v-u}_\infty\ge \ell}$.
    Then, 
    \begin{equation}
        \P(E)\le\left(\frac{\peierls{\P}(\Psi)}{\eps_d}\right)^{\ell+1}.
    \end{equation}
\end{lemma}

\begin{lemma}[{\cite[Lemma 6.6]{hadas2026columnar}}]
    \label{lem:intersection of strongly percolating sets}
    For all $k\ge 1$ and random sets $\Psi_1,\dots,\Psi_k\subset\Z^d$ on a common probability space $(\Omega,\sF,\P)$,
    \begin{equation}
        \peierls{\P}\parentheses[\Bigg]{\bigcap_{i=1}^k\Psi_i}\le k\sqrt[k]{\max_{1\le i\le k}\peierls{\P}(\Psi_i)}.
    \end{equation}
\end{lemma}

\begin{lemma}[{\cite[Lemma 6.7]{hadas2026columnar}}]
    \label{lem:strongly percolating grid events}
    Let $E\subset\Omega$ be an event.
    Let $l_1,\dots,l_d,L_1,\dots,L_d\in\Z_{>0}$ be such that $l_i\mid L_i$ for all $1\le i\le d$.
    Denote $H:=\set{\eta_{z}\mid z\in(l_1\Z\times\dots\times l_d\Z)\cap([0,L_1)\times\dots\times[0,L_d))}$ and $r:=\abs{H}=\frac{L_1\dots L_d}{l_1\dots l_d}$.
    Denote $\Psi:=\xre{l_1\times\dots\times l_d}{\es{E}}$, $\Psi'_\eta:=\xre{L_1\times\dots\times L_d}{\es{\eta E}}$ for each $\eta\in H$, and $\Psi':=\bigcap_{\eta\in H}\Psi'_\eta$.
    Then,
    \begin{enumerate}[label=(\roman*),ref=\thelemma(\roman*)]
        \item \label{itm:strongly percolating grid events - subgrid over grid} $\peierls{\P}(\Psi')\le r\peierls{\P}(\Psi)$.
        \item $\peierls{\P}(\Psi)\le \sqrt[r]{\peierls{\P}(\Psi')}$.
        \item \label{itm:strongly percolating grid events - grid over subgrid}  $\peierls{\P}(\Psi)\le \sqrt[r]{r\sqrt[r]{\max_{\eta\in H}\peierls{\P}(\Psi'_\eta)}}$.
        \item \label{itm:strongly percolating grid events - chessboard bound} If $E\subset\Omega$ is $R$-local for $R=\rect{L_1\times \dots\times L_d}{z}$ and is invariant under reflections through all $d$ axis-aligned hyperplanes passing through the center of $R$, then, for all periodic Gibbs measures $\P$,
        \begin{equation}
            \peierls{\P}(\Psi)\le\sqrt[r]{\norm{\Omega\setminus E}_R}.
        \end{equation}
    \end{enumerate}
\end{lemma}

\begin{lemma}[{\cite[Lemma 2.18]{he2025extended}}]
    \label{lem:strongly percolating grid events - chessboard bound for product measures}
    In the setting of Lemma~\ref{lem:strongly percolating grid events}, if $E\subset\Omega^k$ is $R$-local for $R=\rect{L_1\times \dots\times L_d}{z}$ and is invariant under reflections through the axis-aligned hyperplanes passing through the center of $R$, then, given periodic Gibbs measures $(\P_i)_{i=1}^k$,
        \begin{equation}
            \peierls{\otimes_{i=1}^k\P_i}(\Psi)\le \sqrt[r]{\norm{\Omega^k\setminus E}_R}.
        \end{equation}
\end{lemma}

\begin{lemma}[{\cite[Proposition 6.8]{hadas2026columnar}}]
    \label{lem:splitting strongly percolating sets}
    Let $\P$ be an $\cL$-ergodic probability measure on $\Omega$ for some full-rank sublattice $\cL\subset L_1\Z\times \dots\times L_d\Z$.
    Let $E,F\subset\Omega$ be events.
    If
    \begin{equation}
        \begin{multlined}
            E\cap F=E\cap\eta_{L_i\vec{e}_i}F=\eta_{L_i\vec{e}_i}E\cap F=\emptyset\quad\text{for all }1\le i\le d,
        \end{multlined}
    \end{equation}
    then 
    \begin{equation}
        \min\set{\peierls{\P}(\xre{L_1\times\dots\times L_d}{\es{E}}),\peierls{\P}(\xre{L_1\times\dots\times L_d}{\es{F}})}\le\peierls{\P}(\xre{L_1\times\dots\times L_d}{\es{E\cup F}}).
    \end{equation}
\end{lemma}

\section{The one-dimensional monomer-dimer system}
\label{sec:one-dimensional systems}

In this section, using the transfer matrix method, we compute several quantities related to the partition function of the one-dimensional monomer-dimer system.

The one-dimensional monomer-dimer system considered here is the classical one in which dimers interact only via a hard-core repulsion.
The lateral attraction seen in the higher-dimensional models is instead encoded in a $d$-dependent, \textit{effective} (dimer) chemical potential, $\mu_\eff:=\mu+4(d-1)b$.
To understand this choice, we first observe that, as in~\eqref{eqn:shifted Hamiltonian monomer-dimer}, the one-dimensional model thus defined is equivalently realized by the (formal) shifted Hamiltonian
\begin{equation}
    \label{eqn: shifted Hamiltonian one-dimensional monomer-dimer}
    \frac{1}{2}\mu_\eff\abs{\Z}+H(\sigma)
    =\frac{1}{2}\mu_\eff(\abs{\Z}-2\abs{\sigma}).
\end{equation}
That is, the only nonzero energetic contribution to the shifted Hamiltonian comes from the vacancies, each associated with an additive amount of
\begin{equation}
    \frac{1}{2}\mu_\eff
    =\frac{1}{2}\brackets[\big]{\mu+2(d-1)b}+2(d-1)\cdot\frac{1}{2}b,
\end{equation}
which is the maximum possible energetic penalty associated with a vacancy in the $d$-dimensional model among configurations in which all dimers have the same orientation (cf.~\eqref{eqn:shifted Hamiltonian monomer-dimer}).
Hence, multiplying the one-dimensional partition functions derived here over parallel columns in $d$ dimensions yields lower bounds on the corresponding $d$-dimensional partition functions, restricted to configurations in which all dimers have the same orientation.

\subsection{The model and its exact solution}

We now formally introduce the one-dimensional monomer-dimer system.
Consider $\Z$ in its usual graphical representation as an infinite path, denote its edge set by $\V^{\oned}$, and identify each edge of this path with its midpoint in $\R$.
Given $L\ge 1$, define the one-dimensional box $\mathrm{R}_L:=[-1/2,L-1/2]\subset\R$.
Configurations on $\V^\oned$ and boundary conditions on $\mathrm{R}_L$ are defined as in higher dimensions and denoted using the same notation with the superscript $\oned$ added; as an exception, the finite-volume Hamiltonians and Gibbs measures are adapted from~\eqref{eqn: shifted Hamiltonian one-dimensional monomer-dimer} in the natural way.

To construct the transfer matrix, we picture a sliding window of zero width, centered at each edge of $\V^\oned$, moving from left to right in steps of unit length.
There are exactly two possible states within the window: an edge is either occupied by a dimer or not.
Denote these states by $\ket{\times}$ and $\ket{\circ}$, respectively.
In the $(\ket{\times},\ket{\circ})$-basis, the transfer matrix is
\begin{equation}
    \label{eqn: transfer matrix}
    T=\begin{bmatrix}
        0 & 1 \\
        1 & e^{-\frac{1}{2}\beta\mu_\eff}
    \end{bmatrix}
    ,
\end{equation}
which has eigenvalues
\begin{equation}
    \lambda_1=\frac{1}{2}\parentheses[\Big]{e^{-\frac{1}{2}\beta\mu_\eff}+\sqrt{4+e^{-\beta\mu_\eff}}},
    \quad
    \lambda_2=\frac{1}{2}\parentheses[\Big]{e^{-\frac{1}{2}\beta\mu_\eff}-\sqrt{4+e^{-\beta\mu_\eff}}}.
\end{equation}

Before proceeding, we introduce a mesoscopic length scale that will play a recurrent role in this paper.
Recall that the \textbf{correlation length}~\cite[Section 15.1]{takahashi1999thermodynamics} of the one-dimensional monomer-dimer system is defined as 
\begin{equation}
    \xi:=\frac{1}{\log\abs{\lambda_1/\lambda_2}}=\frac{1}{\log(1+e^{-\frac{1}{2}\beta\mu_\eff}+\order{e^{-\beta\mu_\eff}})}.
\end{equation}
We use the leading-order term of $\xi$ to define the length scale
\begin{equation}
    \ell_0:=e^{\frac{1}{2}\beta\mu_\eff}
    =e^{\frac{1}{2}\beta\brackets{\mu+4(d-1)b}}.
\end{equation}

\subsection{Elementary estimates}
\label{sec:one-dimensional partition function bounds}

First, we derive a lower bound on the partition function on a box of even length with fully packed boundary condition, that is, the boundary condition on $\mathrm{R}_L$ prescribed by the configuration $\rho(e):=\indicator{e-1/2\text{ is even}}$ and denoted by the superscript $1$.

\begin{proposition}
    \label{prop:partition function magnetized boundary}
    For all $\beta>0$ and even $L\ge 4$,
    \begin{equation}
        Z^{1,\oned}_{\mathrm{R}_L;\beta}\ge 1+\frac{1}{8}e^{-\beta\mu_\eff}L^2.
    \end{equation}
\end{proposition}

\begin{proof}
    The fully packed configuration has weight $1$.
    There are $L(L+2)/8$ configurations with exactly two vacancies, each such configuration having weight $e^{-\beta\mu_\eff}$.
    The proposition follows.
\end{proof}

Next, we observe a simple lower bound on the partition function in $d$ dimensions with periodic boundary condition.

\begin{proposition}
    \label{prop:partition function periodic boundary}
    For all boxes $\Lambda$ with at least one side of even length,
    \begin{equation}
        \label{eqn:partition function periodic boundary}
        Z_{\Lambda;\beta}^{\per}\ge 1.
    \end{equation}
\end{proposition}

\begin{proof}
    Without loss of generality, suppose that $\Side_1(\Lambda)$ is even.
    Then, there exists a configuration $\sigma\in\Omega^\per_\Lambda$ in which all dimers are $x_1$-oriented and there are no vacancies.
    Since $w^\per_{\Lambda;\beta}(\sigma)=1$, the proposition follows.
\end{proof}

Finally, we use Proposition~\ref{prop:partition function periodic boundary} to bound the chessboard seminorms of elementary defects in the $d$-dimensional system.

\begin{corollary}
    \label{cor:local defect chessboard seminorms}
    Let $v$ be a vertex and $e$ be an edge of $\Z^d$, and denote by $R_v$ and $R_e$ the $1\times\dots\times 1$ and $1\times\dots\times 1\times 2\times 1\times\dots\times 1$ boxes containing $v$ and $e$, respectively.
    Then,
    \begin{enumerate}
        \item \label{itm:vacancy chessboard seminorm} $\norm{v\text{ is a vacancy}}_{R_v}\le e^{-\frac{1}{2}\beta[\mu+2(d-1)b]}$, and
        \item \label{itm:broken link chessboard seminorm} there exists a constant $C>0$ depending only on $d$ such that $\norm{e\text{ is a broken link}}_{R_e}\le Ce^{-\frac{1}{2}\beta b}$.
    \end{enumerate}
\end{corollary}

\begin{proof}
    Item~\ref{itm:vacancy chessboard seminorm} follows immediately from the definition of the chessboard seminorm and Proposition~\ref{prop:partition function periodic boundary}.
    For Item~\ref{itm:broken link chessboard seminorm}, observe that there are a finite number (independent of the orientation of $e$) of configurations on $R_e$ that make $e$ into a broken link, each having $R_e$-chessboard seminorm no greater than $e^{-\frac{1}{2}\beta b}$.
    Item~\ref{itm:broken link chessboard seminorm} then follows by the subadditivity of the chessboard seminorm.
\end{proof}

\section{Orientational symmetry breaking}
\label{sec: orientational symmetry breaking}

\subsection{The Heilmann--Lieb approach}
\label{sec: Heilmann-Lieb approach}

In this subsection, we recall and recast the Heilmann--Lieb characterization of orientational order in the monomer-dimer system~\cite[Section~5]{heilmann1979lattice} under the strong-percolation framework of Section~\ref{sec:strongly percolating sets}.

Consider the hypercube $R=\mathrm{R}_{2\times \dots\times 2}$ of side length $2$.
We say that $R$ is \textbf{bad} with respect to a configuration $\sigma$ if either (a) it contains a vacancy of $\sigma$, or (b) it does not contain any vacancy but intersects dimers in $\sigma$ of different orientations, and denote by $Q$ the event that $R$ is bad.
It is easy to check by a case analysis that, in the latter case, at least one broken link of $\sigma$ is local to $R$, namely, there exists an edge $e\in\V$ such that the event that $e$ is a broken link is $R$-local.
Combining Proposition~\ref{prop: infinite-volume recursive chessboard estimate} and Corollary~\ref{cor:local defect chessboard seminorms} and applying a union bound, we find that there exists a universal constant $C_d>0$ depending only on $d$ such that
\begin{equation}
    \label{eqn: bad square chessboard seminorm}
    \norm{Q}_R < C_d e^{-\beta\min\set{\mu+2(d-1)b,\,b}/2}.
\end{equation}

If $R$ is not bad, then every vertex of $\Z^d$ in $R$ is incident to a dimer and these dimers all have the same orientation $o\in\set{x_1,\dots,x_d}$.
Denote these $R$-local events by $E_{x_i}$ for $1\le i\le d$, respectively, and introduce the random sets
\begin{equation}
    \Psi_{o}:=\es{E_{o}}=\xre{1\times\dots\times 1}{\es{E_{o}}}\quad\text{for }o\in\set{x_1,\dots,x_d},\quad\text{and }
    \Psi:=\bigcup_{o\in\set{x_1,\dots,x_d}}\Psi_{o}.
\end{equation}
Note that, by compactness, there exists a $\Z^d$-invariant Gibbs measure $\P$ which is the weak limit of a sequence of finite-volume Gibbs measures with periodic boundary condition.
By Lemma~\ref{itm:strongly percolating grid events - chessboard bound}, for all periodic Gibbs measures $\P$,
\begin{equation}
    \peierls{\P}(\Psi)\le \sqrt[2^d]{\norm{Q}_R}.
\end{equation}
On the other hand, notice that if $v\in\Psi_{o}$ and $w\in\Psi_{o'}$ for $o\ne o'$, then $v$ and $w$ are not $\square$-adjacent.
Thus, Lemma~\ref{lem:splitting strongly percolating sets} yields, for all $\Z^d$-ergodic Gibbs measures $\P$,
\begin{equation}
    \min_{o\in\set{x_1,\dots,x_d}}\peierls{\P}(\Psi_{o})\le\peierls{\P}(\Psi).
\end{equation}
From here on, we increase $\beta_0$ so that $\eps:=\sqrt[2^d]{C_d e^{-\beta_0\min\set{\mu+2(d-1)b,b}/2}}<\eps_d/d$, and continue to assume that $\beta\ge\beta_0$.
Then, for all $\Z^d$-ergodic Gibbs measures $\P$, 
\begin{equation}
    \label{eqn: at least one good event strongly percolates}
    \min_{o\in\set{x_1,\dots,x_d}}\peierls{\P}(\Psi_{o}) <\eps<\eps_d/d.
\end{equation}

\begin{proposition}[$d$ phases]
    \label{prop: d phases}
    Suppose that $\beta\ge\beta_0$.
    Let $\P$ be a $\Z^d$-ergodic Gibbs measure.
    There exists a unique label $o\in\set{x_1,\dots,x_d}$ such that, out of the random sets $\Psi_{x_i}$, where $1\le i\le d$, exactly $\Psi_o$ is $\eps$-strongly percolating, while the others have only finite $\boxtimes$-components $\P$-a.s.
\end{proposition}

\begin{proof}
    Let $\P$ be a $\Z^d$-ergodic Gibbs measure.
    In view of~\eqref{eqn: at least one good event strongly percolates}, we may assume without loss of generality that $\peierls{\P}(\Psi_{x_1})\le\eps$, as the other cases are analogous.
    We need to show that $\Psi_{x_i}$, $2\le i\le d$, have finite $\boxtimes$-components $\P$-a.s.
    Fix $\eps'\in(\eps,\eps_d/d)$.
    By~\eqref{eqn: strong percolation parameter} and Lemma~\ref{lem:Peierls argument}, there exists an $\eps'$-rare set $B$ such that $\Psi_{x_1}$ contains the unique infinite $\square$-component $I$ of $\Z^d\setminus B$, and each $\boxtimes$-component of $\Z^d\setminus I$ is finite, $\P$-a.s.
    Noting that $\bigcup_{i=2}^d\Psi_{x_i}\subset\Z^d\setminus\Psi_{x_1}\subset\Z^d\setminus I$ completes the proof.
\end{proof}

For a $\Z^d$-ergodic Gibbs measure $\P$, let $\Phase{\P}\in\set{x_1,\dots,x_d}$ denote the label for $\P$ guaranteed by Proposition~\ref{prop: d phases}.
Let $\varphi:\R^d\to\R^d$ be the coordinate permutation $(x_1,\dots,x_{d-1},x_d)\mapsto(x_d,x_1,\dots,x_{d-1})$, which naturally induces an automorphism of the label set $\set{x_1,\dots,x_d}$.
It is immediate that if $\Phase{\P}=o$, then $\Phase{\varphi^j\P}=\varphi^j\cdot o$ for all $j\in\Z$.

\subsection{Microscopic rarity of misaligned dimers in the $x_1$-phase}
\label{sec: microscopic rarity}

In this subsection, we strengthen the Heilmann--Lieb characterization of the $\Z^d$-ergodic Gibbs measures by establishing the strong percolation of a different random set that, intuitively, captures the microscopic rarity of \textit{misaligned} dimers in the $x_1$-phase; here, we say that a dimer is \textbf{misaligned} if it is not $x_1$-oriented.
By permuting the coordinates, an analogous result holds for each of the other phases as well.

Recall that we have been working under the assumptions that $b>0$, $\mu+2(d-1)b>0$, and $(4d-6)b>\mu$.
Introduce the \textit{$\beta$-independent} microscopic length scales
\begin{equation}
    \label{eqn: microscopic length scale m}
    \fm:=\max\set*{2^{d+1}(4d-5),\,\ceil{\frac{d(4d-5) 2^{d+2} b}{\mu+2(d-1)b}}},
\end{equation}
\begin{equation}
    \label{eqn: microscopic length scale n}
    \fn:=\max\set*{4(6d-5)(\fm+1),\,2\floor{\frac{2(4d-5)b}{(4d-6)b-\mu}}+2}.
\end{equation}
Consider the hypercube $R=\mathrm{R}_{\fn\times\dots\times \fn}$ of side length $\fn$.
Let $\fv$ denote the vertex $(\fn/2,\dots,\fn/2)\in\Z^d$, and define the $R$-local event
\begin{equation}
    \label{eqn: weak vertical order}
    F_{x_1}:=\set{\sigma\in\Omega\mid\fv\text{ is not incident to a misaligned dimer in }\sigma}
\end{equation}
and the random set $\Theta_{x_1}:=\xre{\fn\times\dots\times \fn}\es{F_{x_1}}$.
The main result of this subsection is as follows.

\begin{theorem}
    \label{thm: microscopic rarity of misaligned dimers}
    There exists a constant $C_{\ref{thm: microscopic rarity of misaligned dimers}}>0$ such that the following holds.
    If $\beta\ge\beta_0$ and $\P$ is a $\Z^d$-ergodic Gibbs measure such that $\Phase{\P}=x_1$, then
    \begin{equation}
        \label{eqn: microscopic rarity of misaligned dimers}
        \peierls{\P}(\Theta_{x_1})\le C_{\ref{thm: microscopic rarity of misaligned dimers}} e^{-\beta(1-1/\fn)(4d-5)b}.
    \end{equation}
\end{theorem}

The rest of this subsection is dedicated to the proof of Theorem~\ref{thm: microscopic rarity of misaligned dimers}.
We will proceed by explicitly constructing an $\order{e^{-\beta(1-1/\fn)(4d-5)b}}$-rare set $B$ such that $\Theta_{x_1}$ contains the unique infinite $\square$-component of $\Z^d\setminus B$, $\P$-a.s.

\subsubsection{Construction of the rare set}
\label{sec: construction of the rare set}

In this sub-subsection, we construct the random set $B$.
Let $\P$ be a $\Z^d$-ergodic Gibbs measure such that $\Phase{\P}=x_1$.
The following construction is valid $\P$-a.s.

Sample $\sigma\sim\P$.
Let $V(\sigma)\subset\Z^d$ consist of those vertices which are either (a) a vacancy of $\sigma$, or (b) incident to a misaligned dimer in $\sigma$.
By Proposition~\ref{prop: d phases}, $\peierls{\P}(\Psi_{x_1})\le\eps$, so $V(\sigma)$ has only finite $\boxtimes$- and hence $\square$-components.
Denote by $\cC(\sigma)$ the set of those $\square$-components $C$ of $V(\sigma)$ with the following property: there exists $z\in\Z^d\setminus\Theta_{x_1}(\sigma)$ such that the vertex $\eta_{\fn z}\fv$ belongs to $C$.
For each $C\in\cC(\sigma)$, define the $\fn$-coarse-grained version of $C$ as 
\begin{equation}
    D_C:=\set{z\in\Z^d\mid \eta_{\fn z}\mathrm{R}_{\fn\times\dots\times\fn}\cap C\ne\emptyset}.
\end{equation}
Since $C$ is finite, so is $D_C$, and we may let $I_C$ be the unique infinite $\square$-component of $\Z^d\setminus D_C$ and $H_C:=\Z^d\setminus I_C$ be the hull of $D_C$.
Define
\begin{equation}
    B_C:=\set{z\in H_C\mid\text{there exists }w\in I_C\text{ that is $\square$-adjacent to }z},
\end{equation}
and let $B=B(\sigma):=\bigcup_{C\in\cC(\sigma)}B_C$.
We record the following simple topological property of $B$.

\begin{lemma}
    \label{eqn: rare set topological property}
    If $\Z^d\setminus B(\sigma)$ has a unique infinite $\square$-component $I$, then $I\subset\Theta_{x_1}(\sigma)$.
\end{lemma}

\begin{proof}
    Under the assumption of the lemma, we must have that $I\subset \Z^d\setminus\bigcup_{C\in\cC(\sigma)}H_C$, so
    \begin{equation}
        \bigcup_{C\in\cC(\sigma)}H_C\subset\Z^d\setminus I.
    \end{equation}
    On the other hand, it follows from the definition of $\cC(\sigma)$ that 
    \begin{equation}
    \Z^d\setminus\Theta_{x_1}(\sigma)\subset\bigcup_{C\in\cC(\sigma)}D_C\subset\bigcup_{C\in\cC(\sigma)}H_C.
    \end{equation}
    Combining these inclusions completes the proof.
\end{proof}

\subsubsection{Defects associated with the rare set}
\label{sec: defects associated with the rare set}

We continue to work in the setting of Section~\ref{sec: construction of the rare set}.
To quantify the rarity of the random set $B=B(\sigma)$ constructed there, we will associate a set of defects with each $\square$-component $C\in\cC(\sigma)$ in such a way that distinct $\square$-components have disjoint sets of associated defects.

Given $e\in\V$ and a configuration $\sigma$, we say that $e$ is a \textbf{defect edge} of $\sigma$ if either (a) $e$ is a broken link of $\sigma$, or (b) $e$ is not a broken link of $\sigma$ but is incident to a vacancy of $\sigma$.

Let $\sigma$ be a configuration for which $B(\sigma)$ is defined and $C\in\cC(\sigma)$.
Since $C$ is finite, we may let $I^\mic_C$ be the unique infinite $\square$-component of $\Z^d\setminus C$ and $H^\mic_C:=\Z^d\setminus I^\mic_C$ be the hull of $C$.
Define
\begin{equation}
    \gamma_C:=\set{\set{u,v}\in\V\mid u\in H^\mic_C,\, v\in I^\mic_C,\, u\text{ is $\square$-adjacent to }v}.
\end{equation}
By the definition of $V(\sigma)$ of which $C$ is a $\square$-component, every $e\in\gamma_C$ is a defect edge, and $\gamma_{C}\cap\gamma_{C'}=\emptyset$ for distinct $\square$-components $C,C'\in\cC(\sigma)$.

\begin{lemma}
    \label{lem: a boundary edge in every boundary box}
    For each $z\in B_C$, there exists an edge in $\gamma_C$ that intersects $\eta_{\fn z}\mathrm{R}_{\fn\times\dots\times\fn}$.
\end{lemma}

\begin{proof}
    Let $w\in I_C$ be $\square$-adjacent to $z$.
    Then, $\eta_{\fn w}\mathrm{R}_{\fn\times\dots\times\fn}\cap\Z^d\subset I_C^\mic$.
    Fix $u\in C\cap\eta_{\fn z}\mathrm{R}_{\fn\times\dots\times\fn}$ and $v\in\eta_{\fn w}\mathrm{R}_{\fn\times\dots\times\fn}\cap\Z^d$ arbitrarily, and consider a $\square$-path from $u$ to $v$ in $\eta_{\fn z}\mathrm{R}_{\fn\times\dots\times\fn}\cup \eta_{\fn w}\mathrm{R}_{\fn\times\dots\times\fn}$.
    This path contains the desired edge.
\end{proof}

For each $z\in B_C$, fix an edge $e_{C}(z)$ as supplied by Lemma~\ref{lem: a boundary edge in every boundary box}.

For each edge $e=\set{u,v}\in\V$, if $x_i$ is its orientation, let
\begin{equation}
    e^\ast:=\set*{\frac{1}{2}(u+v)+t\mid t_i=0,\, t_j\in\brackets*{-\frac{1}{2},\frac{1}{2}}\text{ for all }j\ne i}\subset\R^d
\end{equation}
be the $(d-1)$-dimensional \textbf{dual plaquette} of $e$. 
We declare two dual plaquettes to be \textbf{adjacent} if they share a $(d-2)$-dimensional face.
Thus, each dual plaquette is adjacent to exactly $6(d-1)$ other dual plaquettes.
Now, given $\gamma_C$, let $\gamma_C^\ast$ consist of the dual plaquettes of the edges in $\gamma_C$; note that $\gamma_C^\ast$ is connected with respect to the adjacency relation defined above~\cite[Section~2]{lebowitz1998improved}.

We partition $\gamma_C^\ast$ into connected subsets of dual plaquettes, which we call \textbf{pieces}, as follows.
\begin{enumerate}
    \item If $\abs{\gamma_C}<\fm$, then we take $\gamma_C^\ast$ itself as the only piece.
    \item If $\abs{\gamma_C}\ge\fm$, by~\cite[Lemma~14]{BalDeBiasio2022} (based on a fact mentioned in~\cite{alon1997edge}), there exists a partition of $\gamma_C^\ast$ into connected subsets $P_1^\ast,\dots,P_t^\ast$ of dual plaquettes such that 
\begin{enumerate}
    \item for all $1\le k\le t-1$, $\fm-1<\abs{P_k^\ast}\le 1+6(d-1)(\fm-1)$, and
    \item $\abs{P_t^\ast}\le 1+6(d-1)(\fm-1)$.
\end{enumerate}
If $\abs{P_t^\ast}\ge\fm$, then we take $P_1^\ast,\dots,P_t^\ast$ as the pieces.
Otherwise, we necessarily have $t\ge 2$.
After relabelling the subsets as needed, we may assume that $P_{t-1}^\ast\cup P_{t}^\ast$ is connected, and take $P_1^\ast,\dots,P_{t-2}^\ast,P_{t-1}^\ast\cup P_{t}^\ast$ as the pieces.
Thus, in both cases, each piece has cardinality at least $\fm$ and at most $(6d-5)\fm$.
\end{enumerate}

Define a map $\pi_C$ from $B_C$ to the set of pieces of $\gamma_C^\ast$ by mapping each $z\in B_C$ to the unique piece containing the dual plaquette $e_{C}(z)^\ast$.

\begin{lemma}
    \label{lem: cardinality of piece preimage}
    The preimage of each piece of $\gamma_C^\ast$ under $\pi_C$ has cardinality at most $2^d$.
\end{lemma}

\begin{proof}
    Since $\fn\ge 4(6d-5)(\fm+1)$ and the cardinality of a piece is at most $(6d-5)\fm$, each piece intersects the $\fn\times\dots\times\fn$ hypercube $\eta_{\fn z}\mathrm{R}_{\fn\times\dots\times\fn}$ for at most $2^d$ distinct $z\in\Z^d$.
\end{proof}

The main result of this sub-subsection is as follows.

\begin{proposition}[Localization of defect edges]
    \label{prop: localization of defect edges}
    Let $\sigma\in\Omega$ be a configuration for which $B(\sigma)$ is defined and $A\subset\Z^d$ be a nonempty, finite set such that $A\subset B(\sigma)$.
    For each $C\in\cC(\sigma)$ such that $A_C:=A\cap B_C\ne\emptyset$, the following hold:
    \begin{enumerate}
        \item \label{itm: short pieces} If $\abs{\gamma_C}<\fm$, then $B_C=\set{z}$ for some $z\in\Z^d$, and, counting with multiplicity, at least $2(4d-5)$ broken links of $\sigma$ are incident to $C$, all local to $\eta_{\fn z}\mathrm{R}_{\fn\times\dots\times\fn}$.
        \item \label{itm: long pieces} If $\abs{\gamma_C}\ge\fm$, then $\gamma_C$ contains at least $\fm\abs{A_C}/2^d$ distinct defect edges of $\sigma$ (that is, not counting multiplicity for the broken links), all local to
        \begin{equation}
            \bigcup_{z\in A_C}\eta_{(\fn(z_1-1),\dots,\fn(z_d-1))}\mathrm{R}_{3\fn\times\dots\times 3\fn}.
        \end{equation}
    \end{enumerate}
    In particular, the defect edges identified above are distinct for distinct $\square$-components $C$ with $A_C\ne\emptyset$.
\end{proposition}

\begin{proof}
    Let $C\in\cC(\sigma)$ be such that $A_C\ne\emptyset$.

    First, suppose that $\abs{\gamma_C}<\fm$.
    By the definition of $\cC(\sigma)$, $C$ contains the translated vertex $\eta_{\fn z}\fv$, which is incident to a misaligned dimer $e$ in $\sigma$, for some $z\in\Z^d$.
    Since $\fn\ge 4(6d-5)(\fm+1)$, this forces $B_C=D_C=\set{z}$.
    In particular, every edge incident to $C$ has both endpoints contained in $\eta_{\fn z}\mathrm{R}_{\fn\times\dots\times\fn}$; it follows from this observation that all the broken links identified below are local to this hypercube.
    Let $x_i$ ($i\ne 1$) be the orientation of $e$.
    \begin{enumerate}
        \item \label{itm: defect chasing step 1} Consider the maximal run of consecutive vertices in $C$, in the $x_i$-direction passing through $e$.
        The vertex $v_k$ just beyond each endpoint $u_k$ of this run, $k=1,2$, is necessarily incident to an $x_1$-oriented dimer in $\sigma$.
        Hence, the $x_i$-oriented edge $\set{u_k,v_k}$ is a broken link of $\sigma$.
        This step yields two distinct broken links of $\sigma$.

        \item \label{itm: defect chasing step 2} For each $j\not\in\set{1,i}$, consider the maximal run of consecutive vertices in $C$, in the $x_j$-direction passing through either vertex of $e$.
        The vertex $v_k$ just beyond each endpoint $u_k$ of this run, $k=1,2$, is necessarily incident to an $x_1$-oriented dimer in $\sigma$.
        Hence, the $x_j$-oriented edge $\set{u_k,v_k}$ is a broken link of $\sigma$.
        This step yields $4(d-2)$ distinct broken links of $\sigma$ with distinct underlying edges from those identified in Step~\ref{itm: defect chasing step 1}.

        \item \label{itm: defect chasing step 3} For each $j\ne i$, consider the maximal run of consecutive vertices incident to an $x_i$-oriented dimer in $\sigma$, in the $x_j$-direction passing through either vertex of $e$.
        The vertex $v_k$ just beyond each endpoint $u_k$ of this run, $k=1,2$, is not incident to an $x_i$-oriented dimer by construction.
        Hence, the $x_j$-oriented edge $\set{u_k,v_k}$ is a broken link of $\sigma$.
        This step yields $4(d-1)$ distinct broken links of $\sigma$.
        In particular, any edge identified in both Steps~\ref{itm: defect chasing step 2} and~\ref{itm: defect chasing step 3} connects a vertex incident to an $x_1$-oriented dimer to one incident to an $x_i$-oriented dimer, so the edge is a broken link of multiplicity $2$.
    \end{enumerate}
    Altogether, counting with multiplicity, we have identified $2(4d-5)$ broken links of $\sigma$ with the desired localization property.

    Suppose now that $\abs{\gamma_C}\ge\fm$.
    By Lemma~\ref{lem: cardinality of piece preimage}, the image of $A_C$ under $\pi_C$ has cardinality at least $\abs{A_C}/2^d$.
    Moreover, every piece in $\pi_C(A_C)$ has cardinality at least $\fm$, with every dual plaquette in the piece crossing a defect edge of $\sigma$, so $\gamma_C$ contains at least $\fm\abs{A_C}/2^d$ distinct defect edges of $\sigma$.
    For the remaining localization property, recall that for each $z\in A_C$, the piece $\pi_C(z)$ contains the dual plaquette $e_C(z)^\ast$, which intersects $\eta_{\fn z}\mathrm{R}_{\fn\times\dots\times\fn}$.
    Since $\fn\ge4(6d-5)(\fm+1)$ and the piece in question has cardinality at most $(6d-5)\fm$, we conclude that all the defect edges of $\sigma$ crossed by the piece are local to $\eta_{(\fn(z_1-1),\dots,\fn(z_d-1))}\mathrm{R}_{3\fn\times\dots\times 3\fn}$.

    The final claim follows by noting that every defect edge selected for a given component $C\in\cC(\sigma)$ is incident to $C$, and distinct components $C,C'\in\cC(\sigma)$ cannot contain adjacent vertices.
\end{proof}

\subsubsection{Proof of Theorem~\ref{thm: microscopic rarity of misaligned dimers}}

By Lemma~\ref{eqn: rare set topological property}, it suffices to show that $B$ is $C_{\ref{thm: microscopic rarity of misaligned dimers}} e^{-\beta(1-1/\fn)(4d-5)b}$-rare for some $\beta$-independent constant $C_{\ref{thm: microscopic rarity of misaligned dimers}}>0$, uniformly in the choice of the $\Z^d$-ergodic Gibbs measure $\P$ such that $\Phase{\P}=x_1$.

Recall the definition of a defect edge from Section~\ref{sec: defects associated with the rare set}.
When the nature of a defect edge is unspecified, we associate with it an energetic penalty of $\delta_d:=\min\set[\big]{\frac{1}{2d}\brackets[\big]{\frac{\mu}{2}+(d-1)b},\frac{b}{2}}$.
Since every vacancy of a configuration $\sigma$ is incident to exactly $2d$ edges of $\Z^d$ and $\delta_d\le b/2$, summing the energetic penalties associated with distinct broken links (with multiplicity) and defect edges (of unspecified nature) in $\sigma$ underestimates the true energetic penalty associated with $\sigma$ via its vacancies and broken links, as computed by the Hamiltonian~\eqref{eqn:finite-volume Hamiltonian}.

Let $A\subset\Z^d$ be a nonempty, finite set. 
Let $\sigma\in\Omega$ be a configuration such that $B(\sigma)$ is defined and $A\subset B(\sigma)$.
Define
\begin{equation}
    \cC_{<,A}(\sigma):=\set{C\in\cC(\sigma)\mid A\cap B_C\ne\emptyset,\,\abs{\gamma_C}<\fm},
\end{equation}
and let $\cE_{<,A}(\sigma)$ be the union of the multisets of broken links of $\sigma$ associated with $C\in\cC_{<,A}(\sigma)$ via Item~\ref{itm: short pieces} of Proposition~\ref{prop: localization of defect edges}.
Thus, counting multiplicity,
\begin{equation}
    \label{eqn: links associated with short pieces cardinality}
    \abs{\cE_{<,A}(\sigma)}=2(4d-5)\abs{\cC_{<,A}(\sigma)}.
\end{equation}
Define $\cC_{\ge,A}(\sigma)$ and $\cE_{\ge,A}(\sigma)$ analogously so that the latter is a set of distinct defect edges of $\sigma$ identified via Item~\ref{itm: long pieces} of Proposition~\ref{prop: localization of defect edges}, with
\begin{equation}
    \label{eqn: links associated with long pieces cardinality}
    \abs{\cE_{\ge,A}(\sigma)}\ge \frac{\fm}{2^d}\sum_{C\in\cC_{\ge,A}(\sigma)}\abs{A\cap B_C}.
\end{equation}
Note that the defect edges in $\cE_{<,A}(\sigma)\cup\cE_{\ge,A}(\sigma)$ are local to $\bigcup_{z\in A}\eta_{(\fn(z_1-1),\dots,\fn(z_d-1))}\mathrm{R}_{3\fn\times\dots\times 3\fn}$.

Next, with $R=\mathrm{R}_{\fn\times\dots\times\fn}$, consider the set of translations of the grid of $R$,
\begin{equation}
    \cG^R:=\set[\big]{\eta_{z}G^R\mid z\in\set{0,\dots,\fn-1}^d}.
\end{equation}
For each $G\in\cG^R$ and $e\in\V$, let $\indicator{G}(e)$ denote the indicator function that both endpoints of $e$ belong to the same cell of $G$; thus, if $\indicator{G}(e)=1$, then the event that $e$ is a defect edge is local to some cell of $G$.
Note the identity
\begin{equation}
    \label{eqn: shift identity}
    \frac{1}{\fn^{d-1}(\fn-1)}\sum_{G\in\cG^R}\indicator{G}(e)=1,\quad\text{for all }e\in\V.
\end{equation}
Hence, there exists $G\in\cG^R$ such that
\begin{align}
    {}&\frac{b}{2}\sum_{e\in\cE_{<,A}(\sigma)}\indicator{G}(e)
    +\delta_d\sum_{e\in\cE_{\ge,A}(\sigma)}\indicator{G}(e)
    {}&&
    \nonumber
    \\
    \ge{}& \frac{\fn^{d-1}(\fn-1)}{\fn^d}\brackets[\bigg]{\frac{b}{2}\abs{\cE_{<,A}(\sigma)}+\delta_d\abs{\cE_{\ge,A}(\sigma)}}
    {}&&\text{having used~\eqref{eqn: shift identity} to choose $G$}
    \nonumber
    \\
    \ge{}& \bigg(1-\frac{1}{\fn}\bigg)\brackets[\bigg]{(4d-5)b\,\abs{\cC_{<,A}(\sigma)}+\frac{\fm\delta_d}{2^d}\sum_{C\in\cC_{\ge,A}(\sigma)}\abs{A\cap B_C}}
    {}&&\text{using~\eqref{eqn: links associated with short pieces cardinality} and~\eqref{eqn: links associated with long pieces cardinality}}
    \nonumber
    \\
    \ge{}& \bigg(1-\frac{1}{\fn}\bigg)(4d-5)b\brackets[\bigg]{\abs{\cC_{<,A}(\sigma)}+\sum_{C\in\cC_{\ge,A}(\sigma)}\abs{A\cap B_C}}
    {}&&\text{by the choice of $\fm$ in~\eqref{eqn: microscopic length scale m}}
    \nonumber
    \\
    \ge{}& \bigg(1-\frac{1}{\fn}\bigg)(4d-5)b\abs{A}.
    {}&&
    \label{eqn: penalty retainment}
\end{align}

Now, for each $G\in\cG^R$, let $\cF_A(G)$ be the set of cells of $G$ that intersect
\begin{equation}
    \bigcup_{z\in A}\eta_{(\fn(z_1-1),\dots,\fn(z_d-1))}\mathrm{R}_{3\fn\times\dots\times 3\fn}.
\end{equation}
Since any given hypercube of side length $3\fn$ intersects at most $5^d$ cells of $G$, we have that
\begin{equation}
    \label{eqn: number of cells intersecting A}
    \abs{\cF_A(G)}\le 5^d\abs{A}.
\end{equation}
For $\sigma\in\Omega$ and $F\in\cF_A(G)$, define
\begin{align}
    b_F(\sigma){}&:=\abs{\set{e\in\V\mid\text{both endpoints of $e$ belong to $F$, and $e$ is a broken link of $\sigma$}}},
    \\
    d_F(\sigma){}&:=\abs{\set{e\in\V\mid\text{both endpoints of $e$ belong to $F$, and $e$ is a defect edge of $\sigma$}}},
\end{align}
in both of which we count broken links with multiplicity.
Note that $0\le b_F(\sigma)\le d_F(\sigma)\le 2d\fn^{d-1}(\fn-1)$.
For integers $r,s\ge 0$ with $r+s\le 2d\fn^{d-1}(\fn-1)$, define the $F$-local event
\begin{equation}
    E_F(r,s):=\set{\sigma\in\Omega\mid b_F(\sigma)\ge r,\,d_F(\sigma)\ge r+s}.
\end{equation}
Also, define $\cP_A(G)$ as the set of functions $\cF_A(G)\to\set{(r,s)\mid r,s\ge 0,\, r+s\le 2d\fn^{d-1}(\fn-1)}$ such that, writing $(r_F,s_F)$ for the image of $F\in\cF_A(G)$ under such a function, we have
\begin{equation}
    \sum_{F\in\cF_A(G)}\bigg(\frac{b}{2}\cdot r_F+\delta_d\cdot s_F\bigg)\ge \bigg(1-\frac{1}{\fn}\bigg)(4d-5)b\abs{A}.
\end{equation}
By~\eqref{eqn: penalty retainment}, we have the cover
\begin{equation}
    \set{A\subset B}
    \subseteq\bigcup_{G\in\cG^R}\bigcup_{\phi\in\cP_A(G)}\bigcap_{F\in\cF_A(G)}E_F(r_F,s_F).
\end{equation}

It is straightforward to verify that there exists a constant $C>0$, depending only on $\fn$ (which is fixed since we have fixed $\mu$ and $b$), such that, uniformly in $\eta\in\set{\eta_{z}\mid z\in\set{0,\dots,\fn-1}^d}$, the cell $F$ of the translated grid $\eta G^R$, and $r,s\ge 0$,
\begin{equation}
    \norm{E_F(r,s)}_{F}
    \le C\exp{-\beta\bigg(\frac{b}{2}\cdot r+\delta_d\cdot s\bigg)}.
\end{equation}
Recalling~\eqref{eqn: number of cells intersecting A} and using Proposition~\ref{prop: infinite-volume chessboard estimate}, we conclude the proof by a union bound.

\section{Absence of translational symmetry breaking}
\label{sec: absence of translational symmetry breaking}

In this section, we prove the uniqueness of Gibbs measures exhibiting each phase identified in Section~\ref{sec: Heilmann-Lieb approach}. 
As a consequence, we rule out translational symmetry breaking and obtain estimates on the decay of correlations within each phase, completing the proof of Items~\ref{itm:main - invariance and extremality} and~\ref{itm:main - decay of correlations} of Theorem~\ref{thm:main}.
The main result of this section is as follows.

\begin{proposition}
    \label{prop:unique ergodic Gibbs measure for each phase}
    For each $o\in\set{x_1,\dots,x_d}$, there exists a unique $\Z^d$-ergodic Gibbs measure $\P_o$ such that $\Phase{\P_o}=o$. 
    Moreover,
    \begin{enumerate}
        \item \label{itm:unique ergodic Gibbs measure for each phase - extremal} $\P_o$ is extremal;
        \item \label{itm:unique ergodic Gibbs measure for each phase - reflection} for all $2\le j\le d$, $\P_{x_j}$ is the push-forward of $\P_{x_1}$ under $\varphi^{j-1}$, where $\varphi$ denotes the coordinate permutation $(x_1,\dots,x_{d-1},x_d)\mapsto(x_d,x_1,\dots,x_{d-1})$;
        \item \label{itm:unique ergodic Gibbs measure for each phase - decomposition} every periodic Gibbs measure is a convex combination of the $\P_{x_i}$, $1\le i\le d$.
    \end{enumerate}
\end{proposition}

The proof of Proposition~\ref{prop:unique ergodic Gibbs measure for each phase} is by a custom implementation of the method of cluster swapping introduced by Sheffield~\cite{sheffield2005random}, which generalizes that of disagreement percolation introduced by van den Berg~\cite{van1993uniqueness} and enhanced in van den Berg--Maes~\cite{van1994disagreement}.
Rather than relying on the latter method which requires considering an extended connectivity on $L(\Z^d)$ that makes the interacting monomer-dimer models into a Markov random field, here we tailor the clusters to be swapped based on a structural simplification special to this family of models, which in turn yields much more manageable clusters.

In Section~\ref{sec:cluster swapping}, we construct the clusters and state and prove the consequences of cluster swapping under the assumption that the clusters are almost surely finite.
In Section~\ref{sec:clusters are almost surely finite}, we verify that the clusters are indeed almost surely finite.
Lastly, in Section~\ref{sec:deduction of uniqueness of Gibbs measures}, we deduce Proposition~\ref{prop:unique ergodic Gibbs measure for each phase}.

\subsection{Cluster swapping}
\label{sec:cluster swapping}

In this subsection, we specify the clusters to be used in the cluster swapping argument and state and prove the implications of the latter.

We start by defining the clusters. 
Let $\sigma,\sigma'\in\Omega$.
Define $\Delta_{\sigma,\sigma'}$ as the union of the symmetric difference $\sigma\triangle\sigma'$ and the set of edges $e=\set{u,v}\in\V$ such that, letting $x_i$ be its orientation, there exists $j\ne i$ such that
\begin{equation}
    \brackets[\bigg]{\sum_{\substack{f\in\V,\,f\ni u\\x_j\text{-oriented}}}\parentheses[\big]{\sigma(f)-\sigma'(f)}}
    \brackets[\bigg]{\sum_{\substack{f\in\V,\,f\ni v\\x_j\text{-oriented}}}\parentheses[\big]{\sigma(f)-\sigma'(f)}}
    \ne 0.
\end{equation}
A \textbf{cluster} is a (maximal) connected component of $\Delta_{\sigma,\sigma'}$ with respect to the usual adjacency relation on $L(\Z^d)$ (see Section~\ref{sec:basic definitions}).

Suppose that $A\subset\V$ is finite.
Denote by $C_A(\sigma,\sigma')$ the union of all clusters that intersect $A$.
Denote by $C_A'(\sigma,\sigma')$ the union of $A$ and the set of edges in $\V$ (a) at distance one from $A$ (with respect to the adjacency relation on $\V$) or (b) at distance one or two from $C_A(\sigma,\sigma')$.
Define a mapping $m:\Omega^2\to\Omega^2$ (see Lemma~\ref{itm: cluster swapping - hard-core compatibility}), $(\sigma,\sigma')\mapsto(\omega,\omega')$, where
\begin{equation}
    \label{eqn: cluster swapping - mapping}
    \omega(e),\omega'(e):=
    \begin{cases}
        \sigma'(e),\sigma(e) & \text{if }C_A(\sigma,\sigma')\text{ is finite and }e\in C_A(\sigma,\sigma')\\
        \sigma(e),\sigma'(e) & \text{otherwise}
    \end{cases}
    .
\end{equation}
For $A,B\subset\V$, let $A\leftrightarrow B$ denote the event that a cluster intersects both $A$ and $B$.

We can now state the main result of this subsection.

\begin{theorem}[Cluster swapping]
    \label{thm:cluster swapping}
    Let $\P,\P'$ be Gibbs measures of the model. 
    Let $\sigma,\sigma'$ be independent samples from $\P$ and $\P'$, respectively.
    If 
    \begin{equation}
        \label{eqn:cluster swapping - no infinite component}
        (\P\otimes\P')(\Delta_{\sigma,\sigma'}\text{ has an infinite connected component})=0,
    \end{equation}
    then the following hold:
    \begin{enumerate}
        \item $\P=\P'$ and $\P$ is extremal;
        \item Let $A,B\subset\V$ and $f,g:\Omega\to[-1,1]$ be $A$- and $B$-local functions, respectively.
        Then,
        \begin{equation}
            \abs{\Cov_\P(f,g)}\le 2(\P\otimes\P)(A\leftrightarrow B).
        \end{equation}
    \end{enumerate}
\end{theorem}

We start with some elementary properties of the clusters and the mapping $m$.

\begin{lemma}
    Let $A\subset\V$ be finite and $\sigma,\sigma'\in\Omega$, and write $(\omega,\omega')=m(\sigma,\sigma')$. 
    The following hold:
    \begin{enumerate}[label=(\roman*),ref=\thelemma(\roman*)]
        \item \label{itm: cluster swapping - hard-core compatibility} The mapping $m$ is compatible with the hard-core repulsion: $\omega,\omega'\in\Omega$.
        \item \label{itm: cluster swapping - invariance} The clusters are invariant under $m$: $\Delta_{\omega,\omega'}=\Delta_{\sigma,\sigma'}$. Consequently, $C_A(\omega,\omega')=C_A(\sigma,\sigma')$, and $m(\omega,\omega')=(\sigma,\sigma')$, namely, $m$ is an involution.
        \item \label{itm: cluster swapping - potential preservation} Recall the potential notation $\Phi_B$ from Section~\ref{sec:the Hamiltonian}. For all finite $B\subset\V$, $m$ preserves the combined potential at $B$: $\Phi_B(\omega)+\Phi_B(\omega')=\Phi_B(\sigma)+\Phi_B(\sigma')$.
    \end{enumerate}
\end{lemma}

\begin{proof}[Proof of~\ref{itm: cluster swapping - hard-core compatibility}]
    It suffices to show that no swapped dimer violates the hard-core constraint.
    Without loss of generality, suppose that $C_A(\sigma,\sigma')$ is finite and $e\in C_A(\sigma,\sigma')$ is a dimer in $\sigma$.
    Suppose that $f\in\V$ intersects $e$ and is a dimer in $\sigma'$.
    We must have that $f\ne e$, as otherwise $e\not\in\Delta_{\sigma,\sigma'}$.
    However, this implies that $e,f\in\sigma\triangle\sigma'$ and so $f\in C_A(\sigma,\sigma')$.
    Therefore, $f$ is swapped into $\omega$ and cannot intersect $e$ after the latter is swapped into $\omega'$.
\end{proof}

\begin{proof}[Proof of~\ref{itm: cluster swapping - invariance}]
    There is nothing to prove if $C_A(\sigma,\sigma')$ is infinite, so we assume otherwise.
    It is clear that $\omega\triangle\omega'=\sigma\triangle\sigma'$.
    Let $e\in\V$ be an arbitrary edge, $x_i$ be its orientation, and $u$ be an endpoint of $e$.
    It suffices to show that, for each $j\ne i$, the sums
    \begin{equation}
        \sum_{\substack{f\in\V,\,f\ni u\\x_j\text{-oriented}}}\parentheses[\big]{\sigma(f)-\sigma'(f)}
        \quad\text{and}\quad
        \sum_{\substack{f\in\V,\,f\ni u\\x_j\text{-oriented}}}\parentheses[\big]{\omega(f)-\omega'(f)}
    \end{equation}
    are either equal or differ only in sign.
    To this end, consider the set $F$ of $x_j$-oriented edges $f\in\V$ that are incident to $u$ and such that $\sigma(f)\ne\sigma'(f)$. 
    If $F=\emptyset$, then both sums vanish.
    If $F=\set{f}$, then the sums are equal if $f\not\in C_A(\sigma,\sigma')$ and differ only in sign otherwise.
    Otherwise, $F=\set{f_1,f_2}$.
    Since both edges are incident to $u$, either both are in $C_A(\sigma,\sigma')$ or neither one is; the sums are equal in the latter case and differ only in sign in the former.
\end{proof}

\begin{proof}[Proof of~\ref{itm: cluster swapping - potential preservation}]
    There is nothing to prove if $C_A(\sigma,\sigma')$ is infinite, so we assume otherwise.
    It suffices to check the identity for those $B\subset\V$ realizing the vacancy and broken link potentials (see Section~\ref{sec:the Hamiltonian} for the notational conventions used below).
    For the vacancy potential~\eqref{eqn:vacancy potential}, we write
    \begin{equation}
        \Phi_B(\sigma)=\frac{1}{2}\brackets[\big]{\mu+2(d-1)b}\brackets[\bigg]{1-\sum_{\substack{e\in\V \\ e\ni v}}\sigma(e)}.
    \end{equation}
    The required identity follows from the observation that, by construction,
    \begin{equation}
        \label{eqn: cluster swapping - combined occupancy preservation}
        \sigma(e)+\sigma'(e)=\omega(e)+\omega'(e)\quad\text{for all }e\in\V.
    \end{equation}
    For the broken-link potential~\eqref{eqn:broken-link potential}, we start with a definition: for $\sigma\in\Omega$, $1\le i\le d$, and $v\in\Z^d$, let
    \begin{equation}
        \label{eqn: occupancy number}
        n^\sigma_i(v):=\brackets[\bigg]{\sum_{\substack{e\in\V,\,e\ni v\\x_i\text{-oriented}}}\sigma(e)}\in\set{0,1}
    \end{equation} 
    be the number of $x_i$-oriented dimers in $\sigma$ incident to $v$.
    Now, let $e=\set{u,v}$ be the edge under inspection by the broken-link potential, and let $x_i$ be its orientation.
    We write
    \begin{equation}
        \Phi_B(\sigma)=\frac{b}{2}\sum_{j\ne i}\brackets[\Big]{
            n_j^\sigma(u)+n_j^\sigma(v)-2n_j^\sigma(u)n_j^\sigma(v)
        },
    \end{equation}
    noting that, as expected, the summand is equal to $1$ precisely when exactly one endpoint of $e$ is incident to an $x_j$-oriented dimer in $\sigma$ and otherwise vanishes, for each $j\ne i$.
    By~\eqref{eqn: cluster swapping - combined occupancy preservation}, we have that
    \begin{equation}
        \label{eqn: cluster swapping - combined occupancy preservation for broken link potential}
        n_j^{\sigma}(u)+n_j^{\sigma'}(u)=n_j^{\omega}(u)+n_j^{\omega'}(u)
        \quad\text{and}\quad
        n_j^{\sigma}(v)+n_j^{\sigma'}(v)=n_j^{\omega}(v)+n_j^{\omega'}(v)
        \quad\text{for all }j\ne i.
    \end{equation}
    Thus, it suffices to show that
    \begin{equation}
        \sum_{j\ne i}\brackets[\Big]{n_j^\sigma(u)n_j^\sigma(v)+n_j^{\sigma'}(u)n_j^{\sigma'}(v)}
        =\sum_{j\ne i}\brackets[\Big]{n_j^\omega(u)n_j^\omega(v)+n_j^{\omega'}(u)n_j^{\omega'}(v)}.
    \end{equation}
    By~\eqref{eqn: cluster swapping - combined occupancy preservation for broken link potential} and the elementary identity $2(xy+x'y')=(x+x')(y+y')+(x-x')(y-y')$, it suffices to establish the termwise identity
    \begin{equation}
        \brackets[\Big]{n_j^\sigma(u)-n_j^{\sigma'}(u)}\brackets[\Big]{n_j^\sigma(v)-n_j^{\sigma'}(v)}
        =\brackets[\Big]{n_j^\omega(u)-n_j^{\omega'}(u)}\brackets[\Big]{n_j^\omega(v)-n_j^{\omega'}(v)}
    \end{equation}
    for all $j\ne i$.
    If either factor on the LHS vanishes, it follows from the proof of~\ref{itm: cluster swapping - invariance} that the corresponding factor on the RHS also vanishes.
    Otherwise, we have that $e\in\Delta_{\sigma,\sigma'}$, and the proof of~\ref{itm: cluster swapping - invariance} shows that, in passing from each factor on the LHS to its corresponding factor on the RHS, either both factors are unchanged or both factors are negated, so their product is conserved.
\end{proof}

Given these properties as the essential inputs, the proof of Theorem~\ref{thm:cluster swapping} follows the same lines as that of Peled--Spinka~\cite[Section~2.5.1]{peled2020three} and Hadas--Peled~\cite[Section~8.4]{hadas2026columnar}.
Below, following the proof of~\cite[Lemma~8.8]{hadas2026columnar}, we will prove that the mapping $m$ is measure-preserving.
This turns out to be the only step in the entire argument that relies on special features of the interacting monomer-dimer models considered in this paper.
In other words, once the following lemma is established, the proof of Theorem~\ref{thm:cluster swapping} follows that of~\cite[Theorem 8.2]{hadas2026columnar}, with Lemma~\ref{lem: cluster swapping - measure preservation} replacing~\cite[Lemma~8.8]{hadas2026columnar} and the usual connectivity on $L(\Z^d)$ replacing the $\boxtimes$-connectivity on $\Z^2$.
We omit this remaining argument and refer the reader to~\cite{hadas2026columnar} for details.

\begin{lemma}
    \label{lem: cluster swapping - measure preservation}
    In the setting of Theorem~\ref{thm:cluster swapping}, let $(\sigma,\sigma')$ be sampled from $\P\otimes\P'$.
    Then, $m(\sigma,\sigma')$ has the same distribution as $(\sigma,\sigma')$.
\end{lemma}

\begin{proof}
    Let $\sP\subset\sF\otimes\sF$ consist of:
    \begin{enumerate}
        \item \label{itm: infinite clusters} subevents of the event that $C_A(\sigma,\sigma')$ is infinite, and
        \item \label{itm: finite clusters} events of the form $E=\set[\big]{(\sigma,\sigma')\in\Omega^2\mid \sigma_D=\rho_D,\,\sigma'_D=\rho'_D}$, where $\rho,\rho'\in\Omega$ and $D\subset\V$ is finite with $D\supset C_A(\rho,\rho')\cup C_A'(\rho,\rho')$.
    \end{enumerate}
    The family $\sP$ is a $\pi$-system that generates the product $\sigma$-algebra $\sF\otimes\sF$~\cite{hadas2026columnar}, so it suffices to show that $(\P\otimes\P')(E)=(\P\otimes\P')(m^{-1}(E))$ for all $E\in\sP$.
    If $E$ is of type~\ref{itm: infinite clusters}, $m^{-1}(E)=E$ by the definition~\eqref{eqn: cluster swapping - mapping} of $m$, and there is nothing to prove.
    Otherwise, $E$ is of type~\ref{itm: finite clusters}.
    Write $C:=C_A(\rho,\rho')$, $E_\text{in}:=\set[\big]{(\sigma,\sigma')\mid \sigma_C=\rho_C,\,\sigma'_C=\rho'_C}$, and $E_\text{out}:=\set[\big]{(\sigma,\sigma')\mid \sigma_{D\setminus C}=\rho_{D\setminus C},\,\sigma'_{D\setminus C}=\rho'_{D\setminus C}}$.
    We compute
    \begin{align}
        (\P\otimes\P')(E)
        ={}& \E_{\P\otimes\P'}\brackets[\big]{\E_{\P\otimes\P'}\brackets[\big]{\indicator{E}\mid\sigma_{C^c},\sigma'_{C^c}}}
        {}&& \text{by the tower property}
        \nonumber
        \\
        ={}& \E_{\P\otimes\P'}\brackets[\big]{\indicator{E_\text{out}}(\sigma_{C^c},\sigma'_{C^c}) \, \E_{\P\otimes\P'}\brackets[\big]{\indicator{E_\text{in}}\mid\sigma_{C^c},\sigma'_{C^c}}}
        {}&&
        \nonumber
        \\
        ={}& \E_{\P\otimes\P'}\brackets[\big]{\indicator{E_\text{out}}(\sigma_{C^c},\sigma'_{C^c}) \, \P^{\rho}_{C;\beta}(\rho_C)\, \P^{\rho'}_{C;\beta}(\rho'_C)}
        {}&& \begin{subarray}{l} \text{by the DLR condition} \\ \text{and the definition of }E_\text{out} \end{subarray}.
    \end{align}
    Since $D\supset C_A(\rho,\rho')\cup C_A'(\rho,\rho')$, the defining properties of $E$ completely determine, throughout $E$, the membership in $\Delta_{\sigma,\sigma'}$ of each edge contained in $A$ or at distance at most one from $C$.
    Hence, by Lemma~\ref{itm: cluster swapping - invariance},
    \begin{equation}
         m^{-1}(E)=E_\text{out}\cap\set[\big]{(\sigma,\sigma')\mid \sigma_C=\rho'_C,\,\sigma'_C=\rho_C}.
    \end{equation}
    Hence, an analogous computation shows that
    \begin{equation}
        (\P\otimes\P')(m^{-1}(E))
        = \E_{\P\otimes\P'}\brackets[\big]{\indicator{E_\text{out}}(\sigma_{C^c},\sigma'_{C^c}) \, \P^{\rho}_{C;\beta}(\rho'_C)\, \P^{\rho'}_{C;\beta}(\rho_C)}.
    \end{equation}
    Noting that $m(\rho,\rho')=(\rho'_C\,\rho_{C^c},\rho_C\,\rho'_{C^c})$, recalling the definition of the finite-volume Hamiltonian~\eqref{eqn:finite-volume Hamiltonian}, and using Lemma~\ref{itm: cluster swapping - potential preservation}, we get
    \begin{equation}
        H_{C}(\rho_C\,\rho_{C^c})+H_{C}(\rho'_C\,\rho'_{C^c})
        =H_{C}(\rho'_C\,\rho_{C^c})+H_{C}(\rho_C\,\rho'_{C^c}),
    \end{equation}
    so that 
    \begin{equation}
        \P^{\rho}_{C;\beta}(\rho_C)\, \P^{\rho'}_{C;\beta}(\rho'_C)
        =\P^{\rho}_{C;\beta}(\rho'_C)\, \P^{\rho'}_{C;\beta}(\rho_C).
    \end{equation}
    We conclude that $(\P\otimes\P')(E)=(\P\otimes\P')(m^{-1}(E))$, as required.
\end{proof}

\subsection{The clusters are almost surely finite}
\label{sec:clusters are almost surely finite}

In this subsection, we prove that the clusters defined in Section~\ref{sec:cluster swapping} are almost surely finite for sufficiently large $\beta$, thereby verifying the assumption of Theorem~\ref{thm:cluster swapping}.
The characterization follows from the following quantitative statement about the rarity of large clusters, which also enables us to deduce exponential decay of correlations within each phase.

\begin{proposition}
    \label{prop:cluster swapping decay}
    There exist constants $\beta_0,C,c>0$ such that, for all $\beta\ge\beta_0$, if $\P,\P'$ are $\Z^d$-ergodic Gibbs measures with $\Phase{\P}=\Phase{\P'}=x_1$, then, for all $A,B\subset\V$,
    \begin{equation}
        (\P\otimes\P')(A\leftrightarrow B)\le \sum_{u\in A}\sup_{v\in B}\alpha_{x_1}(u,v),
    \end{equation}
    where 
    \begin{equation}
        \label{eqn: Heilmann-Lieb decay}
        \alpha_{x_1}(u,v)
        :=C\exp\set[\bigg]{-c\brackets[\bigg]{\ell_0^{-1}\abs{u_1-v_1}+\sum_{i=2}^d\abs{u_i-v_i}}}.
    \end{equation}
\end{proposition}

\subsubsection{Sealed boxes}
\label{sec:sealed boxes}

Here, we leverage the improved characterization of the $x_1$-phase from Section~\ref{sec: microscopic rarity} to distinguish \textit{sealed} boxes in which a joint configuration $(\sigma,\sigma')\in\Omega^2$ exhibits desirable properties that rule out the existence of large, sprawling clusters nearby.

Introduce a \textit{$\beta$-dependent} mesoscopic length scale
\begin{equation}
    \label{eqn: intermediate length scale}
    \fs:=\fn\ceil{K\ell_0/\fn},
\end{equation}
where $K\ge 1$ is a constant to be chosen later.
Define the events $\Sigma_1\subset\Omega$ and $\Sigma_2,\Sigma_3\subset\Omega^2$ as follows:
\begin{enumerate}
    \item $\sigma\in\Sigma_1$ if and only if every dimer in $\sigma$ that intersects $\eta_{-\fs\vec{e}_1}\mathrm{R}_{3\fs\times 1\times\dots\times 1}$ is $x_1$-oriented.
    \item $(\sigma,\sigma')\in\Sigma_2$ if and only if, for every non-$x_1$-oriented edge $e=\set{u,v}$ having an endpoint in $\eta_{-\fs\vec{e}_1}\mathrm{R}_{3\fs\times 1\times\dots\times 1}$, it holds that (see~\eqref{eqn: occupancy number})
    \begin{equation}
        \label{eqn: no transverse edge}
        \brackets[\Big]{n_1^\sigma(u)-n_1^{\sigma'}(u)}\brackets[\Big]{n_1^\sigma(v)-n_1^{\sigma'}(v)}=0.
    \end{equation}
    \item $(\sigma,\sigma')\in\Sigma_3$ if and only if $\eta_{-\fs\vec{e}_1}\mathrm{R}_{\fs\times 1\times\dots\times 1}$ and $\eta_{\fs\vec{e}_1}\mathrm{R}_{\fs\times 1\times\dots\times 1}$ each contain a vacancy of either $\sigma$ or $\sigma'$.
\end{enumerate}
By an abuse of notation, we write 
\begin{equation}
    \Sigma_1(\sigma):=\Sigma_1\times\Omega
    \quad\text{and}\quad
    \Sigma_1(\sigma'):=\Omega\times\Sigma_1.
\end{equation}
We say that the box $\mathrm{R}_{\fs\times1\times\dots\times 1}$ is \textbf{sealed} in $(\sigma,\sigma')$ if the event $\Sigma:=\Sigma_1(\sigma)\cap\Sigma_1(\sigma')\cap\Sigma_2\cap\Sigma_3$ holds.
Write $\Pi:=\xre{\fs\times1\times\dots\times 1}{\es{\Sigma}}$.

\begin{lemma}
    \label{lem:sealed boxes - no misaligned dimers}
    There exists a constant $\beta_0>0$ such that, for all $\beta\ge\beta_0$, $K\ge 1$, and $\Z^d$-ergodic Gibbs measures $\P$ with $\Phase{\P}=x_1$, with $C_{\ref{thm: microscopic rarity of misaligned dimers}}>0$ as in Theorem~\ref{thm: microscopic rarity of misaligned dimers},
    \begin{equation}
        \peierls{\P}(\xre{3\fs\times\fn\times\dots\times\fn}{\es{\Sigma_1}})
        \le \fn\sqrt[\fn]{\frac{3C_{\ref{thm: microscopic rarity of misaligned dimers}}}{\fn} e^{-\beta (1-1/\fn)(4d-5)b}\fs}
        .
    \end{equation}
\end{lemma}

\begin{proof}
    Recall the vertex $\fv$ from the context of~\eqref{eqn: weak vertical order}, and let $\tau$ denote the translation of $\R^d$ that maps $\fv$ to the vertex $(-\fs,0,\dots,0)$.
    Then,
    \begin{equation}
        \label{eqn: rewriting of Sigma_1}
        \Sigma_1=\bigcap_{i=0}^{\fn-1} \eta_{i\vec{e}_1} \bigcap_{j=0}^{3\fs/\fn-1}\eta_{j\fn\vec{e}_1}\tau F_{x_1}
        .
    \end{equation}
    By Theorem~\ref{thm: microscopic rarity of misaligned dimers} and Lemma~\ref{itm:strongly percolating grid events - subgrid over grid}, we have that 
    \begin{equation}
        \peierls{\P}\parentheses[\Bigg]{\xre{3\fs\times\fn\times\dots\times\fn}{\es{\bigcap\nolimits_{j=0}^{3\fs/\fn-1}\eta_{j\fn\vec{e}_1}\tau F_{x_1}}}}
        \le \frac{3C_{\ref{thm: microscopic rarity of misaligned dimers}}}{\fn}  e^{-\beta (1-1/\fn)(4d-5)b}\fs.
    \end{equation}
    Thus, by Lemma~\ref{lem:intersection of strongly percolating sets}, 
    \begin{equation}
        \peierls{\P}\parentheses{\xre{3\fs\times\fn\times\dots\times\fn}{\es{\Sigma_1}}}
        \le \fn\sqrt[\fn]{\frac{3C_{\ref{thm: microscopic rarity of misaligned dimers}}}{\fn} e^{-\beta (1-1/\fn)(4d-5)b}\fs},
    \end{equation}
    as required.
\end{proof}

The analysis of $\Sigma_2$ is slightly more involved as it requires a similar defect-chasing argument to that of Section~\ref{sec: microscopic rarity} but for a joint configuration.
We begin with some preliminary definitions.
Consider the box $R=\mathrm{R}_{\fn\times\dots\times\fn}$.
Concentric with $R$ is an elementary hypercube $S$ with vertices $\set{\fn/2-1,\fn/2}^d$.
Let $M\subset\Omega^2$ be the $R$-local event consisting of those $(\sigma,\sigma')$ such that:
\begin{enumerate}
    \item \label{itm: oriented environment} every dimer in $\sigma$ and $\sigma'$ that intersects $R$ is $x_1$-oriented;
    \item \label{itm: adjacent vacancy in oriented environment} there exist $i\ne 1$ and an $x_i$-oriented edge $e=\set{u,v}$ of $S$ such that $u$ and $v$ are each a vacancy in at least one of $\sigma$ and $\sigma'$.
\end{enumerate}
Note that the event $M$ is invariant under reflections through the axis-aligned hyperplanes passing through the center of $R$, so Lemma~\ref{lem:strongly percolating grid events - chessboard bound for product measures} can be applied to it.
We start with a bound on $\norm{M}_R$.

\begin{lemma}
    \label{lem: M is rare}
    There exists a constant $C_{\ref{lem: M is rare}}>0$ such that, for all $\beta\ge 0$,
    \begin{equation}
        \norm{M}_R\le C_{\ref{lem: M is rare}} e^{-\beta[\mu+(4d-5)b]}.
    \end{equation}
\end{lemma}

\begin{proof}
    Let $(\sigma,\sigma')\in M$, and $i\ne 1$ and the $x_i$-oriented edge $e=\set{u,v}$ be as in the definition of $M$.
    It suffices to show that the defining properties of $M$ force a combined energetic penalty of at least $\mu+(4d-5)b$ from \textit{$R$-local} vacancies and broken links of $\sigma$ and $\sigma'$.

    First, suppose that $u$ and $v$ are both vacancies in the same configuration, say, $\sigma$ (the other case is identical).
    \begin{enumerate}
        \item Consider the maximal run of consecutive vacancies of $\sigma$ in $R$, in the $x_i$-direction passing through $e$.
        If this run reaches a vertex in $R$ that is adjacent to $\Z^d\setminus R$ (that is, it reaches the \textit{internal boundary} of $R$), then it contains at least $\fn/2$ $R$-local vacancies, which give the combined energetic penalty of
        \begin{align}
            \frac{\fn}{2}\cdot \frac{1}{2}\brackets[\big]{\mu+2(d-1)b}
            \ge {}& 2(\fm+1)\cdot \frac{1}{2}\brackets[\big]{\mu+2(d-1)b}
            {}&& \text{by the first entry in~\eqref{eqn: microscopic length scale n}}
            \nonumber
            \\
            \ge {}& 2\cdot \frac{1}{2}\brackets[\big]{\mu+2(d-1)b}+(4d-6)\cdot\frac{b}{2}
            {}&& \text{by the second entry in~\eqref{eqn: microscopic length scale m}}
            \nonumber
            \\
            = {}& \mu+(4d-5)b,
            {}&&
        \end{align}
        and we are done.
        Otherwise, the vertex $v_k$ just beyond each endpoint $u_k$ of this run, $k=1,2$, is contained in $R$ and is incident to an $x_1$-oriented dimer in $\sigma$, so the $x_i$-oriented edge $\set{u_k,v_k}$ is an $R$-local broken link of $\sigma$; together, we obtain two distinct $R$-local broken links of $\sigma$.

        \item For each $j\not\in\set{1,i}$, consider the maximal run of consecutive vacancies of $\sigma$ in $R$, in the $x_j$-direction passing through either vertex of $e$.
        The same argument as in the previous step shows that, from each run, we obtain either $\fn/2$ $R$-local vacancies (and we are done) or two distinct $R$-local broken links of $\sigma$.
    \end{enumerate}
    Altogether, we obtain $2+4(d-2)$ $R$-local broken links of $\sigma$.
    Along with the two $R$-local vacancies $u$ and $v$, these defects give a combined energetic penalty of
    \begin{equation}
        2\cdot \frac{1}{2}\brackets[\big]{\mu+2(d-1)b}+(4d-6)\cdot\frac{b}{2}
        = \mu+(4d-5)b,
    \end{equation}
    as required.
    
    Second, suppose that $u$ and $v$ are each a vacancy in a different configuration, say, $u$ is a vacancy in $\sigma$ and $v$ is a vacancy in $\sigma'$ (the other case is identical).
    Here, we consider maximal runs of consecutive vacancies of $\sigma$ and $\sigma'$ in $R$ in each direction $x_j$, $j\ne 1$, passing through $u$ and $v$, respectively.
    As above, either we find $\fn/2$ $R$-local vacancies in one of the runs (and we are done) or we find two distinct $R$-local broken links in each run.
    In the latter case, we obtain a total of two vacancies and $4(d-1)$ distinct broken links of $\sigma$ and $\sigma'$, giving a combined energetic penalty of
    \begin{equation}
        2\cdot \frac{1}{2}\brackets[\big]{\mu+2(d-1)b}+(4d-4)\cdot\frac{b}{2}
        > \mu+(4d-5)b,
    \end{equation}
    which completes the proof.
\end{proof}

\begin{lemma}
    \label{lem:sealed boxes - no transverse edges}
    There exist constants $C,\alpha>0$ such that the following holds.
    Suppose that $\beta\ge\beta_0$, where $\beta_0$ is as in Lemma~\ref{lem:sealed boxes - no misaligned dimers}.
    For all $K\ge 1$ and $\Z^d$-ergodic Gibbs measures $\P,\P'$ with $\Phase{\P}=\Phase{\P'}=x_1$,
    \begin{equation}
        \peierls{\P\otimes\P'}(\xre{3\fs\times\fn\times\dots\times\fn}{\es{\Sigma_2}})
        \le C\parentheses[\Big]{e^{-\beta\min\set{(1-1/\fn)(4d-5)b,\,\mu+(4d-5)b}}\fs}^\alpha.
    \end{equation}
\end{lemma}

\begin{proof}
    Let $\tau$ denote the translation of $\R^d$ that maps $(\fn/2-1,\dots,\fn/2-1)$ to the origin.
    Observe the inclusion
    \begin{equation}
        \brackets[\Bigg]{\bigcap_{z_1=\pm\fn/2}\bigcap_{z_2,\dots,z_d=-\fn/2}^{\fn/2}\eta_{(z_1,\dots,z_d)}\parentheses[\big]{\Sigma_1(\sigma)\cap\Sigma_1(\sigma')}}
        \cap
        \brackets[\Bigg]{\bigcap_{\psi\in\set{\mathrm{id},\eta_{(-1,\dots,-1)}}}\bigcap_{z_1=-\fs}^{2\fs-1} \psi\eta_{z_1\vec{e}_1}\tau M^c} \subset\Sigma_2
        .
    \end{equation}
    Indeed, the first bracket rules out misaligned dimers intersecting the box
    \begin{equation}
        \eta_{(-\fs-\fn/2,-\fn/2,\dots,-\fn/2)}\mathrm{R}_{(3\fs+\fn)\times(\fn+1)\times\dots\times(\fn+1)}.
    \end{equation}
    At the same time, via the negation of Property~\ref{itm: adjacent vacancy in oriented environment} in the definition of $M$, the second bracket ensures that~\eqref{eqn: no transverse edge} holds for each non-$x_1$-oriented edge having an endpoint in the box $\eta_{-\fs\vec{e}_1}\mathrm{R}_{3\fs\times1\times\dots\times 1}$, as required by $\Sigma_2$; note that the alternative scenario corresponding to the negation of Property~\ref{itm: oriented environment} in the definition of $M$ is ruled out by intersection with the first bracket.

    On the one hand, by Lemmas~\ref{lem:sealed boxes - no misaligned dimers} and~\ref{lem:intersection of strongly percolating sets},
    \begin{equation}
        \label{eqn: sealed boxes - no transverse edges - first bound}
        \begin{multlined}
            \peierls{\P\otimes\P'}\parentheses[\Bigg]{\xre{3\fs\times\fn\times\dots\times\fn}{\es{\bigcap\nolimits_{z_1=\pm\fn/2}\bigcap\nolimits_{z_2,\dots,z_d=-\fn/2}^{\fn/2}\eta_{(z_1,\dots,z_d)}\parentheses[\big]{\Sigma_1(\sigma)\cap\Sigma_1(\sigma')}}}}
            \\
            \le 4(\fn+1)^{d-1}\sqrt[4(\fn+1)^{d-1}]{\fn\sqrt[\fn]{\frac{3C_{\ref{thm: microscopic rarity of misaligned dimers}}}{\fn} e^{-\beta (1-1/\fn)(4d-5)b}\fs}}.
        \end{multlined}
    \end{equation}
    On the other hand, we rewrite
    \begin{equation}
        \bigcap_{\psi\in\set{\mathrm{id},\eta_{(-1,\dots,-1)}}} \bigcap_{z_1=-\fs}^{2\fs-1} \psi\eta_{z_1\vec{e}_1}\tau M^c
        =\bigcap_{\psi\in\set{\mathrm{id},\eta_{(-1,\dots,-1)}}} \psi \bigcap_{r_1=0}^{\fn-1} \eta_{r_1\vec{e}_1} \bigcap_{p_1=-\fs/\fn}^{2\fs/\fn-1} \eta_{p_1\fn\vec{e}_1}\tau M^c
        .
    \end{equation}
    By Lemmas~\ref{lem:strongly percolating grid events - chessboard bound for product measures}, \ref{itm:strongly percolating grid events - subgrid over grid}, and~\ref{lem: M is rare}, we have that 
    \begin{equation}
        \peierls{\P\otimes\P'}\parentheses[\Bigg]{\xre{3\fs\times\fn\times\dots\times\fn}{\es{\bigcap\nolimits_{p_1=-\fs/\fn}^{2\fs/\fn-1} \eta_{p_1\fn\vec{e}_1}\tau M^c}}}
        \le \frac{3C_{\ref{lem: M is rare}}}{\fn}  e^{-\beta[\mu+(4d-5)b]}\fs.
    \end{equation}
    Thus, by Lemma~\ref{lem:intersection of strongly percolating sets}, 
    \begin{equation}
        \label{eqn: sealed boxes - no transverse edges - second bound}
        \peierls{\P\otimes\P'}\parentheses[\Bigg]{\xre{3\fs\times\fn\times\dots\times\fn}{\es{\bigcap\nolimits_{\psi\in\set{\mathrm{id},\eta_{(-1,\dots,-1)}}} \bigcap\nolimits_{z_1=-\fs}^{2\fs-1} \psi\eta_{z_1\vec{e}_1}\tau M^c}}}
        \le 2\fn\sqrt[2\fn]{\frac{3C_{\ref{lem: M is rare}}}{\fn} e^{-\beta[\mu+(4d-5)b]}\fs}.
    \end{equation}
    Combining~\eqref{eqn: sealed boxes - no transverse edges - first bound} with~\eqref{eqn: sealed boxes - no transverse edges - second bound} using Lemma~\ref{lem:intersection of strongly percolating sets} completes the proof.
\end{proof}

\begin{lemma}
    \label{lem:sealed boxes - sealing patterns}
    For all $\beta,K>0$ and periodic Gibbs measures $\P,\P'$, with $\fs$ given by~\eqref{eqn: intermediate length scale},
    \begin{equation}
        \begin{multlined}
            \peierls{\P\otimes\P'}(\xre{3\fs\times1\times\dots\times 1}{\es{\Sigma_3\cup(\Sigma_1(\sigma)\cap\Sigma_1(\sigma'))^c}})
            \le 2\parentheses[\bigg]{1+\frac{K^2}{8}}^{-2}.
        \end{multlined}
    \end{equation}
\end{lemma}

\begin{proof}
    Denote $S=\eta_{-\fs\vec{e}_1}\mathrm{R}_{3\fs\times1\times\dots\times 1}$.
    By Lemma~\ref{lem:strongly percolating grid events - chessboard bound for product measures},
    \begin{equation}
        \peierls{\P\otimes\P'}(\xre{3\fs\times1\times\dots\times 1}{\es{\Sigma_3\cup(\Sigma_1(\sigma)\cap\Sigma_1(\sigma'))^c}})\le\norm{(\Sigma_1(\sigma)\cap\Sigma_1(\sigma'))\setminus\Sigma_3}_{S}.
    \end{equation}
    With the goal of applying a union bound to the RHS, we fix an $\fs\times1\times\dots\times 1$ box 
    \begin{equation}
        R\in\set{\eta_{-\fs\vec{e}_1}\mathrm{R}_{\fs\times1\times\dots\times 1},\eta_{\fs\vec{e}_1}\mathrm{R}_{\fs\times1\times\dots\times 1}}
    \end{equation} 
    and study the $R$-local event $E_R\subset\Omega^2$ that in both $\sigma$ and $\sigma'$, every vertex in $R$ is incident to an $x_1$-oriented dimer.
    By dividing $\mathrm{R}_{n!\times\dots\times n!}$ into $(n!)^d/\fs$, $\fs\times1\times\dots\times 1$ boxes with disjoint interiors and applying Proposition~\ref{prop:partition function magnetized boundary} to each box, we get that
    \begin{equation}
        Z^\per_{\mathrm{R}_{n!\times\dots\times n!};\beta}
        \ge \parentheses[\bigg]{1+\frac{\fs^2}{8\ell_0^2}}^{(n!)^{d}/\fs}.
    \end{equation}
    Hence,
    \begin{equation}
        \norm{E_R}_R
        \le \limsup_{n\to\infty}\set[\Bigg]{\frac{2^{2(n!)^{d-1}}}{\brackets[\big]{1+\fs^2/(8\ell_0^2)}^{2(n!)^{d}/\fs}}}^{\fs/(n!)^d}
        =\parentheses[\bigg]{1+\frac{\fs^2}{8\ell_0^2}}^{-2}.
    \end{equation}
    Noting that
    \begin{equation}
        (\Sigma_1(\sigma)\cap\Sigma_1(\sigma'))\setminus\Sigma_3\subset \bigcup_{R\in\set{\eta_{\pm\fs\vec{e}_1}\mathrm{R}_{\fs\times1\times\dots\times 1}}}E_R,
    \end{equation}
    the proof of the lemma is complete after applying Proposition~\ref{prop: infinite-volume recursive chessboard estimate} (noting that the corners of $S$ are on $G^R$) and a union bound over the two choices of $R$.
\end{proof}

We are now ready to prove the strong percolation of sealed boxes in the $x_1$-phase.

\begin{proposition}
    \label{prop:sealed boxes strongly percolate}
    For every $\eps>0$, one may choose $K\ge 1$ and then $\beta_0>0$ such that, for all $\beta\ge\beta_0$ and $\Z^d$-ergodic Gibbs measures $\P,\P'$ with $\Phase{\P}=\Phase{\P'}=x_1$, 
    \begin{equation}
        \peierls{\P\otimes\P'}(\Pi)\le \eps.
    \end{equation}
\end{proposition}

\begin{proof}
    Let $\eps>0$.
    Write
    \begin{equation}
        \Sigma=\Sigma_1(\sigma)\cap\Sigma_1(\sigma')\cap\Sigma_2\cap\left[\Sigma_3\cup(\Sigma_1(\sigma)\cap\Sigma_1(\sigma'))^c\right]
        .
    \end{equation}
    By Lemma~\ref{lem:intersection of strongly percolating sets},
    \begin{equation}
        \begin{multlined}
            \peierls{\P\otimes\P'}(\Pi)
            \le 4\max\set[\Big]{
                \peierls{\P\otimes\P'}(\xre{\fs\times1\times\dots\times 1}{\es{\Sigma_1(\sigma)}}),\,
                \peierls{\P\otimes\P'}(\xre{\fs\times1\times\dots\times 1}{\es{\Sigma_1(\sigma')}}),
                \\
                \peierls{\P\otimes\P'}(\xre{\fs\times1\times\dots\times 1}{\es{\Sigma_2}}),\,
                \peierls{\P\otimes\P'}(\xre{\fs\times1\times\dots\times 1}{\es{\Sigma_3\cup(\Sigma_1(\sigma)\cap\Sigma_1(\sigma'))^c}})
            }^{1/4}.
        \end{multlined}
    \end{equation}

    Let $\delta>0$.
    First, by taking $K\ge 1$ sufficiently large, we may ensure that the bound given by Lemma~\ref{lem:sealed boxes - sealing patterns} is no greater than $\delta$.
    Second, note that
    \begin{align}
        {}& \parentheses*{1-\frac{1}{\fn}}(4d-5)b -\frac{1}{2}\brackets{\mu+4(d-1)b}
        {}&&
        \nonumber
        \\
        = {}& \frac{1}{2}[(4d-6)b-\mu] -\frac{(4d-5)b}{\fn}
        {}&&
        \\
        \ge {}& \frac{1}{2}[(4d-6)b-\mu] -\frac{1}{4}[(4d-6)b-\mu]
        {}&& \text{by the second entry in~\eqref{eqn: microscopic length scale n}}
        \nonumber
        \\
        = {}& \frac{1}{4}[(4d-6)b-\mu] >0 
        {}&&\text{by the assumption }(4d-6)b>\mu.
    \end{align}
    Thus, by increasing $\beta_0$ if necessary, we may also ensure that the bound given by Lemma~\ref{lem:sealed boxes - no misaligned dimers} is no greater than $\delta$.
    It is clear that
    \begin{equation}
        \peierls{\P\otimes\P'}(\xre{\fs\times1\times\dots\times 1}{\es{\Sigma_1(\sigma)}})
        \le \peierls{\P}(\xre{\fs\times1\times\dots\times 1}{\es{\Sigma_1}}),
        \quad
        \peierls{\P\otimes\P'}(\xre{\fs\times1\times\dots\times 1}{\es{\Sigma_1(\sigma')}})
        \le \peierls{\P'}(\xre{\fs\times1\times\dots\times 1}{\es{\Sigma_1}}).
    \end{equation}
    Third, note that
    \begin{equation}
        \mu+(4d-5)b-\frac{1}{2}\brackets{\mu+4(d-1)b}
        \ge \frac{1}{2}[\mu+(2d-2)b]
        >0.
    \end{equation}
    Thus, by increasing $\beta_0$ yet again if necessary, we may moreover ensure that the bound given by Lemma~\ref{lem:sealed boxes - no transverse edges} is no greater than $\delta$.
    Hence, by Lemma~\ref{itm:strongly percolating grid events - grid over subgrid}, we get that
    \begin{equation}
        \begin{multlined}
            \peierls{\P\otimes\P'}(\Pi)
            \le 4\max\set[\Big]{
            \sqrt[3\fn^{d-1}]{3\fn^{d-1}\sqrt[3\fn^{d-1}]{\delta}},
            \sqrt[3]{3\sqrt[3]{\delta}}}^{1/4}.
        \end{multlined}
    \end{equation}
    The proof is complete by taking $\delta$ sufficiently small.
\end{proof}

\subsubsection{Bounding the clusters}

The following deterministic statement clarifies the sense in which sealed boxes rule out the presence of large clusters nearby.

\begin{proposition}
    \label{prop:bounded clusters in sealed boxes}
    Let $(\sigma,\sigma')\in\Omega^2$. 
    Suppose that $S=\eta_{(\fs z_1,z_2,\dots,z_d)}\mathrm{R}_{\fs\times1\times\dots\times 1}$ is sealed in $(\sigma,\sigma')$ for some $z\in\Z^d$.
    Let $e\in\V\cap S$.
    If $e\in\Delta_{\sigma,\sigma'}$, then $e$ is $x_1$-oriented, and the connected component of $e$ in $\Delta_{\sigma,\sigma'}$ is contained in $(\fs(z_1-1)-1/2,\fs(z_1+2)-1/2)\times\set{z_2}\times\dots\times\set{z_d}$.
\end{proposition}

\begin{proof}
    Since $(\sigma,\sigma')\in\eta_{(\fs z_1,z_2,\dots,z_d)}\brackets[\big]{\Sigma_1(\sigma)\cap\Sigma_1(\sigma')\cap\Sigma_2}$, $e$ is necessarily $x_1$-oriented, and the connected component of $e$ in $\Delta_{\sigma,\sigma'}$ is a path consisting entirely of $x_1$-oriented edges.
    This path terminates at the first instance of a vacancy in $\sigma$ or $\sigma'$, which exists in both $\eta_{(\fs (z_1-1),z_2,\dots,z_d)}\mathrm{R}_{\fs\times1\times\dots\times 1}$ and $\eta_{(\fs (z_1+1),z_2,\dots,z_d)}\mathrm{R}_{\fs\times1\times\dots\times 1}$ since $(\sigma,\sigma')\in\eta_{(\fs z_1,z_2,\dots,z_d)}\Sigma_3$.
\end{proof}

\subsubsection{Proof of Proposition~\ref{prop:cluster swapping decay}}

With $\eps=\eps_d/2$, choose $K\ge 1$ and $\beta_0>0$ so that the conclusion of Proposition~\ref{prop:sealed boxes strongly percolate} holds.
Let $\alpha_{x_1}(u,v)$ be defined as in~\eqref{eqn: Heilmann-Lieb decay}, where $C\ge 1$ and $c>0$ will be determined at the end of the proof.

Define $f:\V\to\Z^d$ by $f(e):=(\floor{\frac{e_1}{\fs}},\floor{e_2},\dots,\floor{e_d})$.
Let $(\sigma,\sigma')$ be sampled from $\P\otimes\P'$.
We will rely on the following claim that long paths in $\Delta_{\sigma,\sigma'}$ induce long sequences of non-sealed boxes.

\begin{claim}
    \label{clm:disagreement paths induce box paths}
    Let $u,v\in\V$.
    Suppose that $u-v\not\in[-4\fs,4\fs]\times\set{0}\times\dots\times\set{0}$ and that $u,v$ are connected by a path in $\Delta_{\sigma,\sigma'}$.
    Then, $f(u),f(v)$ are connected by a $\boxtimes$-path in $\Z^d\setminus\Pi$.
\end{claim}

\begin{proof}
    Let $P$ be a path in $\Delta_{\sigma,\sigma'}$ connecting $u$ and $v$.
    The map $f$ sends consecutive \textit{edges} in $P$ (which are elements of $\V$) to either the same vertex or $\boxtimes$-adjacent vertices in $\Z^d$.
    Thus, $f(P)$ contains a $\boxtimes$-walk $Q$ in $\Z^d$ connecting $f(u)$ to $f(v)$.
    It suffices to show that $Q\subset\Z^d\setminus\Pi$.

    First, note that, for each $z\in Q$, there exists $w\in P$ such that $w\in R_z:=\eta_{(\fs z_1,z_2,\dots,z_d)}\mathrm{R}_{\fs\times1\times\dots\times 1}$; this follows directly from the definition of $f$.

    By contradiction, suppose that $Q\cap\Pi\ne\emptyset$.
    Let $z\in Q\cap\Pi$, and choose $w\in P$ as in the previous paragraph.
    Then, $R_z$ is sealed.
    By Proposition~\ref{prop:bounded clusters in sealed boxes}, $w$ is $x_1$-oriented, and the connected component of $w$ in $\Delta_{\sigma,\sigma'}$, which includes $u$ and $v$, is contained in $[w_1-2\fs,w_1+2\fs]\times\set{w_2}\times\dots\times\set{w_d}$, contradicting the assumption that $u-v\not\in[-4\fs,4\fs]\times\set{0}\times\dots\times\set{0}$.
\end{proof}

Let $u\in A$.
By a union bound, it suffices to show that
\begin{equation}
    (\P\otimes\P')(\{u\}\leftrightarrow B)\le \sup_{v\in B}\alpha_{x_1}(u,v).
\end{equation}
Suppose that $u$ is connected by a path in $\Delta_{\sigma,\sigma'}$ to some $v\in B$.
If there exists such a $v$ with $u-v\in[-4\fs,4\fs]\times\set{0}\times\dots\times\set{0}$, then
\begin{equation}
    \alpha_{x_1}(u,v)\ge C\exp{-4c\ell_0^{-1}\fs}.
\end{equation}
Otherwise, by Claim~\ref{clm:disagreement paths induce box paths}, $f(u)$ is connected to $f(v)$ by a $\boxtimes$-path in $\Z^d\setminus\Pi$.
By Proposition~\ref{prop:sealed boxes strongly percolate} and Lemma~\ref{lem:quantitative Peierls argument}, the probability of this event is bounded above by
\begin{equation}
    \begin{multlined}
        (\eps/\eps_d)^{1+\inf_{v\in B}\norm{f(v)-f(u)}_\infty}
        \le (\eps/\eps_d)^{\frac{1}{d}\inf_{v\in B}\left(\fs^{-1}\abs{u_1-v_1}+\sum_{i=2}^d\abs{u_i-v_i}\right)}
        \\
        =\sup_{v\in B}(\eps/\eps_d)^{\frac{1}{d}\left(\fs^{-1}\abs{u_1-v_1}+\sum_{i=2}^d\abs{u_i-v_i}\right)}.
    \end{multlined}
\end{equation}
Choosing $C\ge 1,\,c>0$ so that 
\begin{equation}
    C\exp{-4c\ell_0^{-1}\fs}\ge 1,
    \quad
    (\eps/\eps_d)^{\frac{1}{d}}\le e^{-c},
    \quad\text{and}\quad
    (\eps/\eps_d)^{\frac{1}{d\fs}}\le e^{-c\ell_0^{-1}}
\end{equation}
completes the proof.

\subsection{Deduction of Proposition~\ref{prop:unique ergodic Gibbs measure for each phase}}
\label{sec:deduction of uniqueness of Gibbs measures}

The existence of a $\Z^d$-ergodic Gibbs measure $\P$ with $\Phase{\P}=x_1$ follows from the discussion in Section~\ref{sec: Heilmann-Lieb approach}.
To prove uniqueness, suppose that $\P$ and $\P'$ are two such measures, and let $\sigma,\sigma'$ be independent samples from $\P$ and $\P'$, respectively.
By Proposition~\ref{prop:cluster swapping decay}, the assumption~\eqref{eqn:cluster swapping - no infinite component} of Theorem~\ref{thm:cluster swapping} is satisfied, so $\P=\P'$ and $\P$ is extremal.

Now, suppose that $\P$ is a $\Z^d$-ergodic Gibbs measure with $\Phase{\P}=x_j$, $1\le j\le d$.
As noted at the end of Section~\ref{sec: Heilmann-Lieb approach}, $\Phase{\varphi^{d-j+1}\P}=x_1$.
By the uniqueness and extremality of $\P_{x_1}$ just proven, $\P=\varphi^{j-1}\P_{x_1}$ is also the unique measure with phase label $x_j$ and is extremal.

It remains to prove Item~\ref{itm:unique ergodic Gibbs measure for each phase - decomposition}.
Let $\P$ be a periodic Gibbs measure and $\cL\subset\Z^d$ be a full-rank sublattice such that $\P$ is $\cL$-invariant.
By the ergodic decomposition theorem~\cite[Theorem (14.17)]{georgii2011gibbs}, $\P$ is a convex combination of $\cL$-ergodic Gibbs measures.
Thus, it suffices to show that each $\cL$-ergodic Gibbs measure is one of the $\P_{x_i}$, $1\le i\le d$.
To this end, let $\nu$ be an $\cL$-ergodic Gibbs measure.
Consider the $\Z^d$-invariant measure 
\begin{equation}
    \tilde{\nu}:=\frac{1}{\abs{\Z^d/\cL}}\sum_{\eta\in\Z^d/\cL}\eta\nu.
\end{equation}
By the ergodic decomposition theorem and the uniqueness proven above, $\tilde{\nu}$ is a convex combination of $\P_{x_i}$, $1\le i\le d$.
However, noting that each $\P_{x_i}$, $1\le i\le d$, is $\cL$-ergodic by extremality and $\eta\nu$ is $\cL$-ergodic for each $\eta\in\Z^d/\cL$ since the latter translation preserves $\cL$-ergodicity, we now have two $\cL$-ergodic decompositions of $\tilde{\nu}$.
By the uniqueness of the decomposition, we conclude that each $\eta\nu$ is one of the $\P_{x_i}$, $1\le i\le d$.
In particular, $\nu\in\set{\P_{x_i}\mid 1\le i\le d}$, which completes the proof.

\section{Microscopic characterizations of nematic order}
\label{sec:microscopic characterization of nematic order}

In this section, we derive microscopic characterizations of nematic order in terms of the probability of observing an $x_1$-oriented dimer versus a misaligned one within the $x_1$-phase, thus establishing Item~\ref{itm:main - dimer probabilities} of Theorem~\ref{thm:main}.

First, we estimate the probability of observing a misaligned dimer in the $x_1$-phase.

\begin{theorem}
    \label{thm:misaligned dimers are unlikely}
    Suppose that $\beta\ge\beta_0$.
    If $e\in\V$ is not $x_1$-oriented, then
    \begin{equation}
        \P_{x_1}(\sigma(e)=1)=\order{e^{-\beta(1-1/\fn)(4d-5)b}}.
    \end{equation}
\end{theorem}

\begin{proof}
    Let $v\in\Z^d$ be an endpoint of $e$, and recall the event $F_{x_1}$ from~\eqref{eqn: weak vertical order}.
    Since $\P_{x_1}$ is $\Z^d$-invariant, we have that
    \begin{equation}
        \P_{x_1}(\sigma(e)=1)
        \le \P_{x_1}(v\text{ is incident to a misaligned dimer in }\sigma)
        =\P_{x_1}(F_{x_1}^c)
        \le \frac{\peierls{\P_{x_1}}(\Theta_{x_1})}{\eps_d}
    \end{equation}
    by Lemma~\ref{lem:quantitative Peierls argument}.
    Applying Theorem~\ref{thm: microscopic rarity of misaligned dimers} completes the proof.
\end{proof}

Second, we estimate the probability of observing an $x_1$-oriented dimer in the $x_1$-phase.

\begin{theorem}
    Suppose that $\beta\ge\beta_0$.
    If $e\in\V$ is $x_1$-oriented, then
    \begin{equation}
        \P_{x_1}(\sigma(e)=1)=\frac{1}{2}-\order{e^{-\frac{1}{2}\beta[\mu+2(d-1)b]}}.
    \end{equation}
\end{theorem}

\begin{proof}
    Let $e\in\V$ be $x_1$-oriented.
    Let $v\in\Z^d$ be an endpoint of $e$, and let $e'$ be the other $x_1$-oriented edge incident to $v$.
    Consider the events
    \begin{align}
        E_1{}&:=\set{\sigma(e)+\sigma(e')=1}, \\
        E_2{}&:=\set{v\text{ is incident to a misaligned dimer in }\sigma}, \\
        E_3{}&:=\set{v\text{ is a vacancy of }\sigma}.
    \end{align}
    Note that these events partition $\Omega$.
    By Theorem~\ref{thm:misaligned dimers are unlikely}, $\P_{x_1}(E_2)=\order{e^{-\beta(1-1/\fn)(4d-5)b}}$.
    By Proposition~\ref{prop: infinite-volume chessboard estimate} and Item~\ref{itm:vacancy chessboard seminorm} of Corollary~\ref{cor:local defect chessboard seminorms}, $\P_{x_1}(E_3)\le e^{-\frac{1}{2}\beta[\mu+2(d-1)b]}$.
    Hence, 
    \begin{equation}
        \P_{x_1}(E_1)=1-\P_{x_1}(E_2)-\P_{x_1}(E_3)=1-\order{e^{-\frac{1}{2}\beta[\mu+2(d-1)b]}},
    \end{equation}
    having used that $(4d-6)b>\mu$ and $\fn\ge 4(6d-5)[2^{d+1}(4d-5)+1]$ in the second equality.
    Since $\P_{x_1}$ is $\Z^d$-invariant, $\P_{x_1}(E_1)=2\P_{x_1}(\sigma(e)=1)$, which implies the theorem.
\end{proof}

\section*{Acknowledgments}

We thank the anonymous referee for our recent paper~\cite{he2025extended} on Model~I for suggesting the simplified proof strategy adopted in its revision, which enabled the progress on Models~II and~V reported here.

\section*{AI use statement}

The critical observation in Lemma~\ref{itm: cluster swapping - potential preservation} was discovered by GPT-6 Astra in a session dedicated to exploring the applicability of cluster swapping to Model~II.
The same AI model was consulted in extending the proof to arbitrary dimensions $d\ge 2$ and used to program the Monte Carlo simulation leading up to Figure~\ref{fig: model V}.

\bibliographystyle{plain}
\bibliography{bibliography}

\end{document}